\documentclass[]{article}

\usepackage[margin=2.5cm]{geometry}
\usepackage{authblk}
\usepackage{standalone}
\usepackage[hidelinks]{hyperref}
\usepackage{doi}
\usepackage[shortlabels]{enumitem}
\usepackage{physics}
\usepackage{nicematrix}
\usepackage{tikz}
\usetikzlibrary{positioning, arrows.meta, shapes.geometric}

\usepackage{amsthm}
\usepackage{amsmath}
\usepackage{amsfonts}
\usepackage{amssymb}

\theoremstyle{plain}
\newtheorem{theorem}{Theorem}
\newtheorem{lemma}[theorem]{Lemma}
\newtheorem{corollary}[theorem]{Corollary}
\newtheorem{proposition}[theorem]{Proposition}

\theoremstyle{definition}
\newtheorem{definition}{Definition}
\theoremstyle{remark}
\newtheorem{remark}{Remark}

\numberwithin{theorem}{section}
\numberwithin{remark}{section}
\numberwithin{equation}{section}

\newcommand{\email}[1]{\footnotesize\texttt{#1}}

\title{Computational Cryptography from Pseudoentanglement}
\author[1]{Ilia Ryzov${}^*$}
\author[2]{Manuel Goulão${}^*$}
\author[1]{Faedi Loulidi}
\author[1]{David Elkouss}

\affil[1]{{\small Okinawa Institute of Science and Technology Graduate University, Japan\authorcr \email{\{ilia.ryzov,faedi.loulidi,david.elkouss\}@oist.jp} }}
\affil[2]{{\small INESC-ID, Instituto Superior Técnico, Universidade de Lisboa, Portugal\authorcr \email{manuel.goulao@inesc-id.pt}}}

\date{}

\begin{document}

\maketitle

\let\svthefootnote\thefootnote
\let\thefootnote\relax
\footnotetext{$^*$Equal contribution.}
\let\thefootnote\svthefootnote
\setcounter{footnote}{0}

\begin{abstract}
    The advent of pseudoentanglement and computational entanglement theory bootstrapped a wave of research at the intersection of computer science and information theory.
    In parallel, computational cryptography has undergone substantial development, prompted by the introduction of pseudorandom states and followed by the establishment of a baseline for the computational hardness required for quantum cryptography, from which EFI pairs emerge as a central primitive.

    We study the connection between pseudoentanglement and computational cryptography through EFI pairs.
    Our goal is to enable the use of resources arising from computational entanglement theory in the field of cryptography.
    For this, we establish the relation between operational instances of pseudoentanglement and the hierarchy of minimal assumptions for computational cryptography.
    We show that the existence of pseudoentanglement under two different operational definitions, with efficient state generation, is a sufficient condition for the existence of EFI pairs.
    Combined with a previously established result that the converse also holds under the second definition, this allows us to also demonstrate their equivalence.
    This places pseudoentanglement alongside other minimal assumptions in cryptography, not only offering an alternative perspective on this fundamental problem, but also building a bridge that allows insights from either area to inform the other.
    While proving these theorems, we introduce and demonstrate technical lemmas in quantum information and computational entanglement theory, relating the computational entanglement measures to the distance between states, establishing distinguishing conditions for mixtures of two families given pairwise distances between their states, and demonstrating the first continuity relation for a computational entanglement measure.
\end{abstract}

\clearpage
\section{Introduction}
A central question in theoretical cryptography is the study of the computational hardness assumptions that are required for different, more or less powerful, cryptographic constructions and enable their functionalities.
In classical cryptography, a concept on which all computational cryptography (when parties run in Probabilistic Polynomial Time (PPT)) relies is the existence of One-Way Functions (OWFs) --- functions that are efficiently computable in the forward direction but hard to invert.
Yet, the question of the existence of OWFs is known to be hard, as an affirmative answer would imply that P$\neq$NP.
 
However, when one considers quantum cryptography, and parties run in Quantum Polynomial Time (QPT), there exist computational primitives that are weaker than OWFs, i.e., there exist computational hardness assumptions that enable cryptographic constructions that can be built even if OWFs do not exist, as first demonstrated with pseudorandom quantum states~\cite{C:JLS18,AQY22}. 
Particularly relevant for this work is the concept of EFI pairs~\cite{BCQ23} --- families of quantum states that are efficiently preparable and statistically far while computationally indistinguishable. 
EFI pairs are arguably the current best candidate to serve as the minimal assumption in quantum computational cryptography. This is since they are implied by most computational cryptography primitives, such as pseudorandom states, and is equivalent to commitment schemes, oblivious transfer, and secure multiparty computation~\cite{BCQ23}.

In a separate line of research, in~\cite{ABF24,ABV23}, the concept of pseudoentanglement --- families of states with low entanglement that are computationally indistinguishable from families with higher entanglement --- was independently introduced.
While the definition of~\cite{ABF24} first introduced this concept to relate pseudorandom states~\cite{C:JLS18} to physical applications such as matrix product state testing or the AdS/CFT correspondence, the definition of~\cite{ABV23} focuses on establishing an \emph{operational framework} for the study of interactive protocols under Local Operations and Classical Communication (LOCC), and is hence a natural fit for cryptography.
In what follows, we will work in the setting of the second, operational definition, in particular targeting the open question posed in~\cite{ABV23} relating pseudoentanglement and cryptography.

The study of relating pseudoentanglement and computational cryptography was started in~\cite{GE24}, where it is established that an adaptation of the definition of pseudoentanglement, inefficiently-distillable pseudoentanglement, is necessary for the existence of EFI pairs. 
Subsequently,~\cite{GY25} showed that yet another adaptation of the definition of pseudoentanglement, regularized relative pseudoentanglement, constitutes a sufficient condition for the existence of EFI pairs.
Still, the relationship between the original definition of pseudoentanglement in~\cite{ABV23} and quantum cryptography remained unknown, and given its operational interpretation, it stands as the prime candidate for the definition to work under in the setting of the intrinsic operational field of computational cryptography.
In this work, we address this question and study the relation between the different operational definitions of pseudoentanglement and cryptography.
In particular, we focus on the broad open question:
\begin{center}\em
    What is the relation between pseudoentanglement and cryptography?
\end{center}
We answer this question by exhibiting multiple connections and placing different definitions of pseudoentanglement in relation to EFI pairs (and thus to most computational cryptography).
We present our main conceptual contributions in Sections~\ref{sec: pseudo abv to efi} and~\ref{sec: equiv}, building on technical results presented in Section~\ref{sec: error continuity and gap} --- a continuity bound on the computational entanglement cost, a relation between gaps in computational entanglement and trace distance, and an amplification result for the trace distance between mixtures of states.
Since the goal of this work is to enable the use of operational notions of pseudoentanglement for quantum cryptography, we assume throughout that the families underlying pseudoentanglement are efficiently preparable.
First, in Section~\ref{sec: pseudo abv to efi}, we show that pseudoentanglement as defined in~\cite{ABV23}, which we call \emph{fully-computational pseudoentanglement}, is sufficient to construct EFI pairs, and give an explicit construction.
Second, in Section~\ref{sec: equiv}, we build on~\cite{GE24} and prove the equivalence between their version of pseudoentanglement, which we call \emph{inefficiently-distillable pseudoentanglement}, and EFI pairs.

Our results have strong implications not only in the field of quantum cryptography but also in the new area of computational entanglement theory.
First, our work is the first to connect the pseudoentanglement definition of~\cite{ABV23} to EFI pairs and hence to cryptography.
It offers yet another perspective on the role of EFI pairs as a candidate for minimal assumption.
As observed in~\cite{GY25}, no cryptographic primitive has been constructed directly from the original operational definition of~\cite{ABV23}.
Explicitly, our results imply that if EFI pairs do not exist, then fully-computational pseudoentanglement does not exist.
Given the proposed physical interpretations of pseudoentanglement~\cite{ABV23}, this result provides a potential experimental platform for constructing EFI pairs and other cryptographic primitives from physical phenomena.
Second, combined with the results of~\cite{GE24}, our work opens a new direction to explore computational hardness relations in cryptography within the framework of computational entanglement. Precisely, it gives another completely different but equivalent primitive to EFI pairs, and provides a new resource for building this weaker level of minimal assumptions in computational cryptography. 
We illustrate our main conceptual contributions in Figure~\ref{fig:main work ideas}.

\begin{figure}[tb]
    \centering
    \includegraphics[width=.66\linewidth]{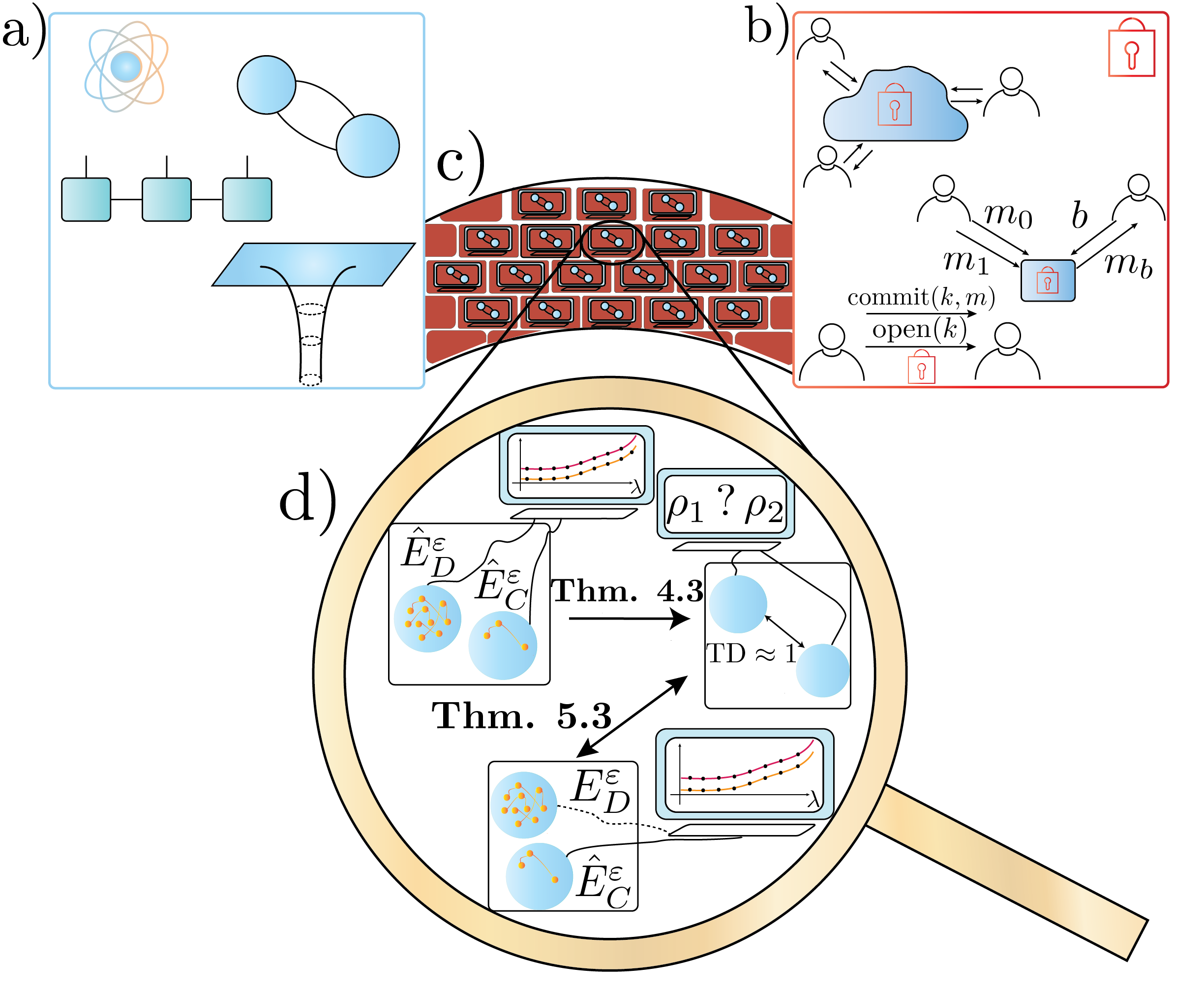}
    \caption{Depiction of the main conceptual results of this work, relating the worlds of physics and computational cryptography. a) Physical world populated by (top to bottom) entanglement, matrix product states, AdS/CFT correspondence; b) Computational cryptography world with (top to bottom) multiparty computation, oblivious transfer and commitments; c) Bridge between physical and computational cryptography worlds; d) Zoom in on the building block of the bridge, with relations between (clockwise from the upper-left corner) fully-computational pseudoentanglement, EFI pairs, and inefficiently-distillable pseudoentanglement.}
    \label{fig:main work ideas}
\end{figure}

\subsection{Contributions}
   In this section, we give an overview of the main contributions of this work.
   We separate our results into two groups: conceptual and technical.
   While the conceptual results are technical in themselves, we believe they have an impact on computational cryptography that goes beyond what is inherent in the mathematical techniques.
   On the same note, we highlight some of our technical results used to prove the main theorems, as they may be of independent interest.
   We highlight our conceptual results as relations in the field of minimal assumptions for cryptography in Figure~\ref{fig:relations}.
    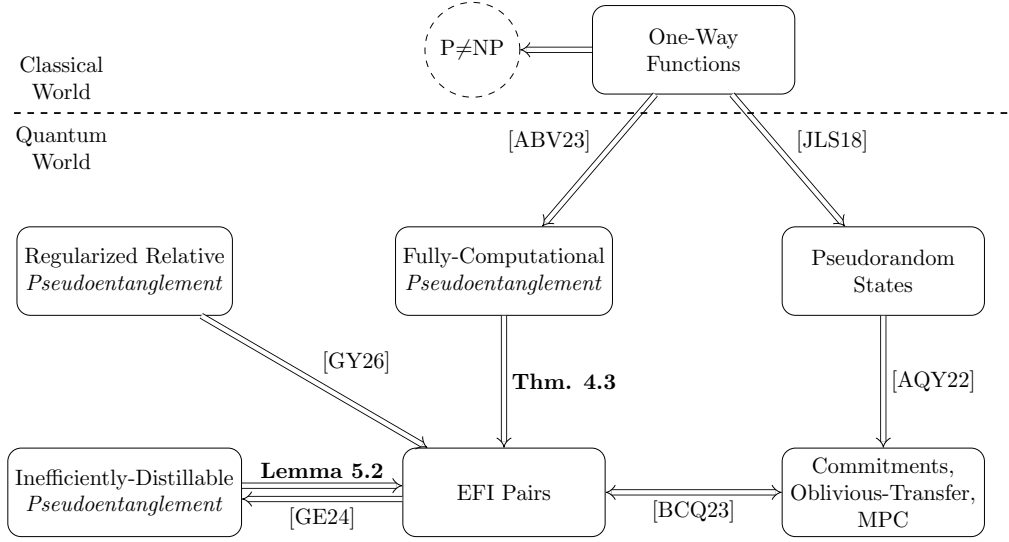
\begin{figure}[tb]
        \centering
        	\begin{tikzpicture}[
		mynode/.style={
			rectangle,
			rounded corners=6pt,
			draw,
			align=center,
			minimum width=3.2cm,
			minimum height=1.4cm
		},
		myarrow/.style={-{Implies}, double, double distance=2pt},
		myarrow2/.style={{Implies}-{Implies}, double, double distance=2pt},
		myarrow3/.style={{Implies}-, double, double distance=2pt}
		]
		
		\node[mynode] (b) at (0,0) {Regularized Relative\\ \textit{Pseudoentanglement}};
		\node[mynode] (c) at (6,0) {Fully-Computational\\ \textit{Pseudoentanglement}};
		\node[mynode] (d) at (12,0) {Pseudorandom\\ States};
		
		\node[mynode] (a) at (9,3.5) {One-Way\\Functions};
		\node[circle, draw, dashed, minimum size=1.5cm] (h) at (5.5,3.5) {P\(\neq\)NP};
		
		\node[mynode] (e) at (0,-3.5) {Inefficiently-Distillable\\ \textit{Pseudoentanglement}};
		\node[mynode] (f) at (6,-3.5) {EFI Pairs};
		\node[mynode] (g) at (12,-3.5) {Commitments,\\ Oblivious-Transfer,\\ MPC};
		
		\draw[myarrow] (a) -- node[above left] {\cite{ABV23}} (c);
		\draw[myarrow] (a) -- node[above right] {\cite{C:JLS18}} (d);
		\draw[myarrow] (b) -- node[above right] {\cite{GY25}} (f);
		\draw[myarrow] (d) -- node[right] {\cite{AQY22}} (g);
		\draw[myarrow] (c) -- node[right] {\textbf{Thm. \ref{thm:pe_to_efi}}} (f);
		
		\draw[myarrow] (a) -- node[] {} (h);
        \draw[myarrow] ([yshift=3pt]e.east) -- node[above] {\textbf{Lemma  \ref{lemma: pse imply efi GE}}} ([yshift=3pt]f.west);
        \draw[myarrow] ([yshift=-3pt]f.west) -- node[below] {\cite{GE24}} ([yshift=-3pt]e.east);		\draw[myarrow2] (f) -- node[below] {\cite{BCQ23}} (g);

    \draw[dashed, thick] (-1.75,2.5) -- (14,2.5);
    \node[above, align=center] at (-1,2.6) {Classical\\ World};    \node[below, align=center] at (-1,2.4) {Quantum\\ World};        	

\end{tikzpicture}
        \caption{Relations of pseudoentanglement with computational cryptography.\protect\footnotemark}         
        \label{fig:relations}
    \end{figure}

\medskip
\noindent\textbf{Conceptual contributions}

    \smallskip
    \textit{Fully-computational pseudoentanglement implies EFI:}
    We show that the existence of EFI pairs is a necessary condition for the existence of pseudoentanglement as introduced in~\cite{ABV23}, presented in Section~\ref{sec: pseudo abv to efi}.
    One of the consequences is that computational cryptography may be constructed from pseudoentanglement, as our proof provides an explicit construction of EFI pairs from pseudoentanglement.
    EFI pairs can then be used to construct other cryptographic primitives, such as commitment schemes, oblivious transfer, or secure multiparty computation~\cite{BCQ23}, and so, one can also efficiently build these from pseudoentanglement.
    Since this connection enables the creation of EFI pairs from physical properties related to computational hardness~\cite{ABV23}, it offers a new path for constructing computational cryptography from emergent physical phenomena.
    Hence, this result further establishes EFI pairs as a basal primitive in the study of minimal assumptions for computational cryptography.

    \smallskip
    \textit{Inefficiently-distillable pseudoentanglement implies EFI:}
    We demonstrate that the existence of inefficiently-distillable pseudoentanglement (as introduced in~\cite{GE24}, extending~\cite{ABV23}) implies the existence of EFI pairs.
    This result, combined with the previous results of~\cite{GE24}, establishes that this version of pseudoentanglement is equivalent to EFI pairs, meaning that it is both necessary and sufficient for their existence.
    This result, presented in Section~\ref{sec: equiv}, has two consequences.
    First, it extends the base level of current minimal assumptions for computational cryptography with a new primitive, pseudoentanglement, that has a different structure, framing, and motivation from the previous primitives at this level.
    Second, it introduces an emerging physical concept into the study of computational hardness, opening new possibilities for connecting other primitives or obtaining impossibility results by bridging known facts from other areas, relevant in either direction (using physics results for cryptography or cryptography results for physics).

\footnotetext{For a broader overview of relations in computational cryptography (albeit without pseudoentanglement), see \url{https://sattath.github.io/microcrypt-zoo/}.}

\medskip
\noindent \textbf{Other technical results}

    \smallskip
    \textit{Error-continuity of $\hat E^\varepsilon_C$:}
    We show that if two families of states are close to each other in trace distance, then if the first family has a valid upper bound on the computational entanglement cost $\hat E^\varepsilon_C$, the other must also have the same valid upper bound.
    Still, the new dilution error of the states becomes a function of the previous one and of the trace distance between the states.
    Moreover, if both the error of the first family and its trace distance from the second family go to zero, then the error of the second family also goes to zero.
    This establishes the first continuity relation on the recently introduced computational entanglement measures.
    These results are stated in Lemma~\ref{lemma: continuity cost to the same epr} and Corollary~\ref{cor: continuity cost uniform}.
    
    \smallskip
    \textit{Computational entanglement vs.\ trace distance:} 
    Building on top of the previous result, we derive a relation between the existence of an entanglement gap between two families of states, and the trace distance between states in those families (Lemma~\ref{lemma: connection of td and cd}).
    We show that if states in two families are close to each other in trace distance, then the entanglement gap quantified by the computational one-shot distillable entanglement $\hat E_D^\varepsilon$ and computational entanglement cost $\hat{E}^\varepsilon_C$ cannot exist. 
    As an immediate consequence, the contrapositive of this argument establishes that the existence of a gap in entanglement implies that the states in the corresponding families are far from each other in trace distance. 

    \smallskip
    \textit{Distinguishing between mixtures of distant states:}
    In another direction, we show that, given polynomially many copies of an unknown (mixed) states whose pairwise trace distances across the two families are non-negligible, one can identify the source family with a probability exponentially close to one in the number of copies (Lemma~\ref{lemma: td of mixtures}).
    This result is consistent with the usual polynomial amplification techniques in cryptography. However, it is directly applicable to mixtures of states, making it potentially useful for tasks such as state discrimination.

\subsection{Technical Overview}
Here, we provide an overview of the central concepts and technical results established in this work. 
We first recall the notions that we deal with in this work: pseudoentanglement and EFI pairs. 
We then provide a summary of our main technical results: 
the error-continuity of computational entanglement cost $\hat E_C^\varepsilon$, the relation between the gap in computational entanglement and the trace distance, and near-perfect distinguishability between uniform mixtures of states.
Finally, we outline our conceptual results, which consist of relations between the existence of different notions of pseudoentanglement, fully-computational pseudoentanglement~\cite{ABV23} and inefficiently-distillable pseudoentanglement~\cite{GE24}, and the existence of EFI pairs.

\medskip\noindent
\textbf{Background concepts}

\smallskip
\textit{Pseudoentanglement:}
Informally, a family of states $\{\psi_{AB}^\lambda\}_\lambda$ is said to be pseudoentangled if it seems to possess entanglement from the perspective of a QPT observer, while not actually having it.
That is, $\{\psi_{AB}^\lambda\}_\lambda$ possesses a ``low'' amount of entanglement (\(c\)), but for an observer  limited to efficient computations, this family is indistinguishable from another family $\{\phi_{AB}^\lambda\}_\lambda$ that has a ``high'' amount of entanglement (\(d\)), with \(c < d\). 
Since the states may be mixed, this notion depends on the choice of entanglement measure, leading to multiple incompatible definitions of pseudoentanglement.
In this work, we focus on two versions of pseudoentanglement that use operational entanglement measures.
In particular, we work with the one-shot entanglement cost --- states can be reconstructed with $c(\lambda)$ EPR pairs and error at most $\varepsilon(\lambda)$ ---, and the one-shot distillable entanglement --- states can be used to produce $d(\lambda)$ EPR pairs with error at most $\varepsilon(\lambda)$ ---, and their computational versions (LOCC is performed by QPT-restricted parties). 
These two versions are:
\begin{itemize}
    \item Fully-computational pseudoentanglement~\cite{ABV23}: two families of states $\{k,\psi_{AB}^k\}_k$, $\{k,\phi_{AB}^k\}_k$, with entanglement quantified by the uniform computational one-shot entanglement cost ($\hat E_C^\varepsilon(\{k,\psi^k_{AB}\}) \leq c$) and uniform computational one-shot distillable entanglement ($\hat E_D^\varepsilon(\{k,\phi^k_{AB}\}) \geq d$).
    \item Inefficiently-distillable pseudoentanglement~\cite{GE24}: two families of states $\{\psi_{AB}^\lambda\}_\lambda$, $\{\phi_{AB}^\lambda\}_\lambda$, with entanglement quantified by the one-shot computational entanglement cost ($\hat E_C^\varepsilon(\{\psi^\lambda_{AB}\}) \leq c$) and information-theoretic one-shot distillable entanglement ($E_D^\varepsilon(\{\phi^\lambda_{AB}\}_\lambda) \geq d$).
\end{itemize}
We note that computational indistinguishability takes essentially the same form in both definitions and states that the states from both families cannot be distinguished by any QPT adversary, even given polynomially many copies.
In~\cite{GE24}, inefficiently-distillable pseudoentanglement was proven to be a necessary condition for EFI pairs. 

\medskip
\textit{EFI pairs~\cite{BCQ23}:}
Two families of states \(\{\rho_0^\lambda\}_\lambda,\{\rho_1^\lambda\}_\lambda\) are an EFI pair if they are 
\begin{itemize}
    \item Efficiently preparable: a uniform QPT algorithm \(\mathcal{A}(1^\lambda, b)\) returns \(\rho_b^\lambda\);
    \item Statistically far: \(\frac{1}{2}\|\rho_0^\lambda - \rho_1^\lambda\|_1\) is noticeable;
    \item Computationally indistinguishable: for all adversaries, \(\{\rho_0^\lambda\}_\lambda \approx \{\rho_1^\lambda\}_\lambda\).
\end{itemize}
EFI pairs are particularly interesting because they can be used to construct various computational cryptographic primitives, such as quantum commitments or secure multiparty computation, and are implied by (i.e., are weaker than) almost all other candidate assumptions (e.g., OWFs, pseudorandom states)~\cite{BCQ23,AQY22}, making them a prime candidate for minimal assumption in quantum cryptography.

\medskip
\noindent\textbf{Main contributions}

\smallskip
\textit{Error-continuity of \(\hat E_C^\varepsilon\) and entanglement gaps:}
Throughout this work, we introduce and use a notion of continuity of the computational entanglement cost, $\hat E_C^\varepsilon$, that we refer to as \emph{error-continuity}. 
Informally, we show that if one can construct states $\{\rho^\lambda\}_\lambda$ from less than $n(\lambda)$ EPR pairs with precision $\varepsilon(\lambda)$, then one may do the same for any family $\{\sigma^\lambda\}_\lambda$ whose states are close to those in $\{\rho^\lambda\}_\lambda$, paying a penalty in the error $\varepsilon'(\lambda)$ that depends on their closeness.
\begin{lemma}[Error-continuity of $\hat E_C^\varepsilon$. Informal Lemma~\ref{lemma: continuity cost to the same epr}]
    Let $\{\rho^\lambda\}_\lambda$ and $\{\sigma^\lambda\}_\lambda$ be two families with states at most $\delta$-close to each other in trace distance, $\frac{1}{2}\|\rho^\lambda - \sigma^\lambda\|_1\leq\delta(\lambda)$, such that $\sqrt{\varepsilon(\lambda)}+\delta(\lambda)\leq 1$.
    Assume that  $\hat E_C^\varepsilon(\{\rho^\lambda\}) \leq n $. Then
    \(
        \hat E_C^{\varepsilon^\prime}(\{\sigma^\lambda\}) \leq n,
    \)
    for $\varepsilon^\prime$ polynomial in $(\delta,\sqrt{\varepsilon})$.
\end{lemma}
This result follows from applying the Fuchs-van de Graaf inequality to translate the condition on the error of dilution of $\hat E_C^\varepsilon$ to a condition on the trace distance between the diluted state and the dilution target.
The same result holds for uniform computational entanglement cost as well.
This error-continuity of $\hat E_C^\varepsilon$ allows us to relate the different pseudoentanglement notions to EFI pairs.
In particular, it is necessary to show that a small enough trace distance between states of two different families implies a vanishing gap between their entanglement, when $\hat E_C^\varepsilon$ and $\hat E_D^\varepsilon$ are used to quantify this gap.
\begin{lemma}[TD and entanglement gap. Informal Lemma~\ref{lemma: connection of td and cd}]\label{td2gap-informal}
    Let $\{\rho^\lambda\}_\lambda$ and $\{\sigma^\lambda\}_\lambda$ be two families of states with $\hat E_C^\varepsilon(\{\rho^\lambda\})\leq c$, $\hat E_D^\varepsilon(\{\sigma^\lambda\})\geq d$, and $\varepsilon$ vanishing at infinity.
    If the trace distance between the states is upper-bounded by a vanishing function \(\delta\), $\frac{1}{2}\|\rho^\lambda - \sigma^\lambda\|_1\leq\delta(\lambda)$, then, for large enough $\lambda$, $c(\lambda)\geq d(\lambda)$.
\end{lemma}
To show this lemma, we use the lower bound $\hat E^\varepsilon_D(\{\sigma^\lambda\})\geq d$ to get a corresponding $\hat{E}_C^\varepsilon(\{\sigma^\lambda\})\geq d$, since 
\( \hat{E}^\varepsilon_C\geq E_C^\varepsilon\geq E^\varepsilon_D\geq \hat E_D^\varepsilon, \)
which holds for any small enough $\varepsilon(\lambda)$.
Then, we consider the error-continuity of $\hat{E}_C^\varepsilon$ to show that
\( \hat E_C^{\varepsilon^\prime}(\{\sigma^\lambda\})\leq c, \)
which, together with $\hat E_C^{\varepsilon^\prime}(\{\sigma^\lambda\})\geq d$ because $\varepsilon^\prime(\lambda)$ and $ \varepsilon(\lambda)$ are both sufficiently small, ultimately yields
\( d\leq\hat E_C^{\varepsilon^\prime}\leq c ,\)
and hence \(c(\lambda) \geq d(\lambda)\).

\medskip
\textit{Distinguishing between mixtures of distant states:}
While taking a mixture of states can render the trace distance to zero, we show that one can amplify the trace distance between mixtures of states by allowing the use of multiple copies of such states.
We consider a binary quantum state discrimination task in which a source emits \(q\) copies of a state \({\rho_b^k}{}^{\otimes q}\) sampled from one of two families, $\{k,\rho_0^k\}_k$ and $\{k,\rho_1^k\}_k$, for \(b\) and \(k\) uniformly random.
Assuming that the pairwise trace distance between states in the families is not negligible, we show that one can discriminate \(b\) up to negligible probability.
\begin{lemma}[Distinguishing between mixtures of distant states. Informal Lemma~\ref{lemma: td of mixtures}]\label{lemma: td mixtures. informal}
    Let $\{k,\rho_0^k\}_k$ and $\{k,\rho_1^k\}_k$ be two families of mixed states and classical keys.
    If the pairwise trace distance between states from different families is lower-bounded by an inverse-polynomial $1/p(\lambda)$, then there exists another polynomial $q(\lambda)$ such that, for large enough $\lambda$ and for a negligible \(\varepsilon\),
    \(
        \frac{1}{2}\left\|\mathbb E_k\left[{\rho_0^k}{}^{\otimes q(\lambda)}\right] -\mathbb E_{k}\left[{\rho_1^{k}}{}^{\otimes q(\lambda)}\right] \right\|_1 \geq 1-\varepsilon(\lambda).
    \)
\end{lemma}
We demonstrate this result using, first, the Holevo-Helstrom bound  (Proposition~\ref{prop: holevo-helstrom}) to translate between the trace distance and the distinguishability of two states and, second, a Pretty Good Measurement (PGM)~\cite{HW94} that provides a sufficiently strong lower bound on the probability of guessing the wrong family.
We then analyze the error probability using Gram matrix techniques~\cite{M07,M19}, and show that the number of copies $q$ required to render the probability of guessing the wrong family negligible can be taken to be polynomial in \(\lambda\).

\medskip
\textit{Fully-computational pseudoentanglement implies EFI pairs:}
We show the first connection between fully-computational pseudoentanglement~\cite{ABV23} and EFI pairs by demonstrating that the existence of the first is a sufficient condition for the existence of the second.
This result tackles an open question in~\cite{ABV23}.
\begin{theorem}[Fully-computational pseudoentanglement $\Rightarrow$ EFI pairs. Informal Theorem~\ref{thm:pe_to_efi}]\label{thm:pe2efi-informal}
    Let $\{k,\psi^k\}_k$ and $\{k,\phi^k\}_k$ be the families in fully-computational pseudoentanglement~\cite{ABV23} (\(c<d\)). Assume that there exist QPT algorithms generating $\psi^k$ and $\phi^k$, given $k$. Then, there exists an EFI pair, \((\rho_0^\lambda,\rho_1^\lambda)\).
\end{theorem}
To prove this theorem, we provide an explicit construction of an EFI pair \((\rho_0^\lambda,\rho_1^\lambda)\), using states from the two pseudoentanglement families, $\{k,\psi^k\}_k$ , $\{k,\phi^k\}_k$, given the assumption that they are efficiently preparable.
Our construction takes uniform mixtures of polynomially many copies of states in the respective families:
\begin{equation*}
    \left(\rho_0^\lambda,\rho_1^\lambda\right)_\lambda :=\left(\mathbb E_k\left[{\psi^k}^{\otimes q(\lambda)}\right], \mathbb E_k\left[{\phi^k}^{\otimes q(\lambda)}\right]\right)_\lambda.
\end{equation*}
For the proposed construction, we verify the requirements of an EFI pair:
\begin{itemize}
    \item Efficient preparation is given by assumption. We remark that while this assumption is not present in the definition of pseudoentanglement, it is natural when one wants to use the states in an efficient protocol.
    This property is verified in all current constructions of pseudoentanglement~\cite{ABV23,GE24,ABF24,BF24}.
    \item Computational indistinguishability follows directly from the definition of pseudoentanglement, as it is the same statement in both definitions.
    \item Statistical distance requires a more complex reasoning. We leverage the connection between trace distance and entanglement gap (Lemma~\ref{td2gap-informal}) to ensure that the two elements of the EFI pair, each corresponding to the uniform mixture over the keys of a different family, are arbitrarily far. 
\end{itemize}

\medskip
\textit{Inefficiently-distillable pseudoentanglement is equivalent to EFI pairs:}
We then demonstrate that the existence of inefficiently-distillable pseudoentanglement~\cite{GE24} is both necessary and sufficient for the existence of EFI pairs.
While the first condition has been shown in~\cite{GE24}, we prove the reverse direction, which we intuitively explain next.

\begin{lemma}[Inefficiently-distillable pseudoentanglement $\Rightarrow$ EFI pairs. Informal Lemma~\ref{lemma: pse imply efi GE}]
    Let $\{\psi^\lambda\}_\lambda$ and $\{\phi^\lambda\}_\lambda$ be the families in inefficiently-distillable pseudoentanglement families~\cite{GE24} (\(c<d\)). 
    Assume there exist uniform QPT algorithms generating $\psi^\lambda$ and $\phi^\lambda$.
    Then, there exists an EFI pair, \((\rho_0^\lambda,\rho_1^\lambda)\).
\end{lemma}
We provide an explicit construction of EFI pairs from pseudoentanglement,
\begin{equation*}
    \left(\rho_0^\lambda,\rho_1^\lambda\right)_\lambda :=(\psi^\lambda, \phi^\lambda)_\lambda.
\end{equation*}
We leverage Lemma~\ref{td2gap-informal}, which connects trace distance and entanglement gap, to show that these states in the candidate EFI pair are far from each other in trace distance; otherwise, it could not be that $c < d$ (i.e., the families would not have pseudoentanglement).
The remaining required properties of the EFI pair follow in the same way as above in Theorem~\ref{thm:pe2efi-informal}.

Combined with the results of~\cite{GE24}, we get that the existence of inefficiently-distillable pseudoentanglement is equivalent to the existence of EFI pairs.
\begin{theorem}[Inefficiently-distillable pseudoentanglement  \(\Leftrightarrow\) EFI pairs. Informal Theorem~\ref{thm:ipe_equiv_efi}]
   The existence of inefficiently-distillable $(\varepsilon,c,d)$-pseudo\-entanglement --- with $c < d$, $\varepsilon(\lambda)\in O(2^{-\lambda})$, and both families efficiently preparable --- is both necessary and sufficient for the existence of EFI pairs.
\end{theorem}

\subsection{Related Work}
We provide an overview of related work aimed at connecting different notions of pseudoentanglement and cryptography.
To date, multiple constructions of pseudoentanglement of mixed and pure states have been proposed~\cite{ABF24,ABV23,GB23}.
Still, these exhibit a dependence on OWFs, which are not known to exist and are too strong for quantum cryptography.

Our work builds on and advances a line of research~\cite{ABV23,GE24,GY25}, originating from an open question raised in~\cite{ABV23} regarding the study of connections between pseudoentanglement and computational cryptography.
In this line, the use of pseudoentanglement as a basis for computational cryptography has been explored in two works~\cite{GE24,GY25}, which we analyze below.
We further progress this study by establishing new links and completing open relations between computational hardness and physical resources under operational settings.

\smallskip
\textit{EFI pairs imply inefficiently-distillable pseudoentanglement~\cite{GE24}.}
In this work, a construction of pseudoentanglement is proposed that relies only on EFI pairs. 
To show this, first, a new definition of pseudoentanglement (inefficiently-distillable pseudoentanglement) is introduced, which is an adaptation of the definition in~\cite{ABV23} (fully-computational pseudoentanglement). 
The main difference is that the uniform computational distillable-entanglement condition is replaced by its one-shot information-theoretic counterpart.
Intuitively, it reflects the idea that entanglement in the reference family need not be efficiently accessible and need only exist, motivated by the fact that this family may not even be efficiently preparable.

This definition of pseudoentanglement enables the construction of two computationally indistinguishable families of states, one of which is efficiently preparable by LOCC channels without shared entanglement, therefore $\hat E_C^0 = 0$, while, at the same time, the other possesses one ebit of distillable entanglement (with negligible distillation error), $E_D^{2^{-\lambda}}=1$.
Then, it is shown that an entanglement gap can be amplified by taking any polynomial number of copies of the aforementioned states without affecting computational indistinguishability.

Lastly, inefficiently-distillable pseudoentanglement is proposed as a new minimal assumption for the existence of computational cryptography, while providing an explicit construction of pseudoentanglement from EFI pairs with both families efficiently preparable.
The same work raises the question of whether the reverse implication holds, i.e., whether the existence of inefficiently-distillable pseudoentanglement implies the existence of EFI pairs.

\smallskip
\textit{Pseudoresources imply EFI pairs~\cite{GY25}.}
While pseudoentanglement is the most well-established computational physical
resource, other physical pseudoresources have also been introduced, such as
pseudomagic and pseudocoherence. The recent work of~\cite{GY25} studies the
relation between general pseudoresources and computational cryptography through
EFI pairs. To this end, the authors introduce an alternative notion of mixed-state
pseudoentanglement --- regularized relative pseudoentanglement --- and show that it
implies the existence of EFI pairs. This notion quantifies entanglement by the
asymptotic i.i.d.\ regularized relative entropy of
entanglement, rather than by the one-shot distillable entanglement and
entanglement cost of~\cite{ABV23,GE24}. The asymptotic measure offers
structural properties, notably continuity, that are useful for constructing
EFI pairs. However, it is not tied to efficient LOCC protocols, and
hence does not directly address the operational settings that motivated the
original definition~\cite{ABV23}.

Our objective is complementary, and we aim to enable the use of pseudoentanglement in
quantum cryptography, and therefore we work with operational notions such as
those of~\cite{ABV23,GE24}. Their advantage lies in the computational
accessibility of the underlying entanglement, since applications that involve
entanglement manipulation require the relevant operations to be efficiently
implementable. Fully-computational and inefficiently-distillable
pseudoentanglement guarantee a separation in the entanglement that the parties
can actually operate over (for preparation and distillation in the former case,
and for preparation only in the latter), whereas regularized relative
pseudoentanglement does not by construction provide such a guarantee.
The work of~\cite{GY25} shows that
computational cryptography can be built from pseudoentanglement even when the
underlying entanglement need not be efficiently accessible to the parties, 
establishing the connection from an information-theoretic perspective. We, 
in turn, establish this connection from operational entanglement measures, which support
concrete applications such as cryptographic protocols. 
This also has a bearing on protocol-level verifiability, as our operational
setting admits explicit witnesses, with the relevant keys and
descriptions of the efficient distillation and dilution maps, for any arbitrary entanglement gap. 
These can provide a blueprint for certifying entanglement in cryptographic and LOCC settings. By
contrast, a gap in the regularized relative entropy of entanglement does not by
itself yield manipulation or testing procedures of this kind, and require a gap above $2 + 1/\mathrm{Poly}(\lambda)$.

The two works are therefore complementary, and both advance the connection
between quantum resources and computational hardness: \cite{GY25} strengthens
the evidence for an intrinsic link between pseudoentanglement and computational
cryptography through an information-theoretic lens, while our results extend
this insight to operational scenarios, including the original definition~\cite{ABV23}.

\section{Background}
We first describe some common objects and notation that we will use. We denote natural numbers by $\mathbb N$, the real numbers as $\mathbb R$, and the complex numbers by \(\mathbb{C}\).
We denote the first $N$ natural numbers by $[N]:=\{1,2,\dots,N\}$. 
We mainly work with quantum states in a finite-dimensional Hilbert space $\mathcal{H}\cong \mathbb C^d$.
A quantum state $\rho$ is given by a self-adjoint, unit trace, positive semidefinite operator. The space of all quantum states on an underlying Hilbert space $\mathcal{H}$ is denoted as $\mathcal D(\mathcal H):=\{\rho\in\mathcal L(\mathcal H): \operatorname{Tr}(\rho)=1,\,\rho\geq0,\,\rho^\dagger=\rho\}$, where $\mathcal L(\mathcal H)$ is the space of linear operators acting on the Hilbert space $\mathcal H$.
A particular class of states is the class of pure states, which are given as rank-1 projectors in the form $\psi = \ketbra{\psi}{\psi}$, where $\ket{\psi}$ is a normalized vector in the Hilbert space $\mathcal H$. 
A relevant pure bipartite (across \(A:B\)) state is the maximally entangled state or EPR pair, $\Phi_{AB}:= \ketbra{\varphi}{\varphi}_{AB}$ with $\ket{\varphi}_{AB}:=\frac{1}{\sqrt{2}}(\ket{00}_{AB} + \ket{11}_{AB})$.
The most general description of the transformations of states is given in the form of completely positive, trace-preserving superoperators, often referred to as a quantum channel.
We denote the identity operator acting on underlying Hilbert space $\mathcal H$ as $\mathbb I$. 

\smallskip

We now introduce formal definitions of the concepts required in the remainder of this work.
First, we introduce measures that allow us to compare different quantum states.
For a more in-depth review of quantum information theory, we direct the reader to~\cite{W18,KW24}.

\begin{definition}[Schatten \(p\)-norm]\label{def:shattenpnorm}
The \emph{Schatten \(p\)-norm} of a linear operator $X\in\mathcal L(\mathcal H)$ is defined for $1\leq p<+\infty$ as 
\begin{equation*}
    \|X\|_p:=\operatorname{Tr}(|X|^p)^{\frac{1}{p}},
\end{equation*}
where $|X|:=\sqrt{X^\dagger X}$. \\
Special cases of a Schatten norm that are relevant for this work are the \emph{1-norm} (trace norm) and \emph{2-norm} (Frobenius norm).  The trace norm of $X\in\mathcal L(\mathcal H)$ is given by
\begin{equation*}
\|X\|_1=\operatorname{Tr}\left(\sqrt{X^\dagger X}\right)=\sum_{i=1}^rs_i,
\end{equation*}
where $r=\operatorname{rank}(X)$ and $\{s_i\}_{i=1}^r$ are the singular values of $X\in\mathcal L(\mathcal H)$. The Frobenius norm of an operator $X$ is given by 
\begin{equation*}
    \|X\|_2 = {\operatorname{Tr(X^\dagger X)}}^{\frac{1}{2}} = \left(\sum_{i,j=1}^{n}|X_{ij}|^2\right)^{\frac{1}{2}},
\end{equation*}
where $n = \operatorname{dim}\mathcal H$ and $X_{ij}$ are the matrix elements of $X$.
\end{definition}

\begin{definition}[Fidelity] \label{def: fidelity}
Let $\rho,\sigma\in \mathcal D(\mathcal{H})$ be two quantum states. 
The \emph{fidelity} is defined as
\begin{equation*}
    \operatorname{F}(\rho,\sigma) := \|\sqrt{\rho}\sqrt{\sigma}\|_1^2,
\end{equation*}
where $\|\cdot\|_1$ is the Schatten 1-norm.  
\end{definition}

\begin{definition}[Trace Distance (TD)]
Let $\rho,\sigma\in \mathcal D(\mathcal{H})$ be two quantum states. 
The \emph{trace distance} is defined as
\begin{equation*}
    \operatorname{TD}(\rho,\sigma):=\frac{1}{2}\|\rho-\sigma\|_1,
\end{equation*}
where $\|\cdot\|_1$ is the Schatten 1-norm.  
\end{definition}
The trace distance and fidelity possess a number of useful properties.
We recall some that are frequently utilized throughout the work.
\begin{proposition}[Multiplicativity of Fidelity]\label{prop: mult fidelity}
    Let $\rho_1,\sigma_1\in\mathcal D(\mathcal H_1)$ and $\rho_2,\sigma_2\in\mathcal D(\mathcal H_2 )$ be two pairs of quantum states. Then,
    \begin{equation*}
        \operatorname{F}(\rho_1\otimes\rho_2,\sigma_1\otimes\sigma_2) = \operatorname{F}(\rho_1,\sigma_1)\operatorname{F}(\rho_2,\sigma_2).
    \end{equation*}
\end{proposition}

\begin{proposition}[Fuchs-van de Graaf inequality~\cite{FvdG99}] \label{prop: fvdg}
Let $\rho,\,\sigma \in \mathcal D(\mathcal{H})$ be two quantum states. Then,
\begin{equation*}
    1-\sqrt{\operatorname{F}(\rho,\sigma)} \leq \operatorname{TD}(\rho,\sigma) \leq\sqrt{1-\operatorname{F}(\rho,\sigma)}.
\end{equation*}
\end{proposition}

The ability to distinguish between any two quantum states $\rho,\sigma\in \mathcal D(\mathcal{H})$ by the best possible measurement is related to the trace distance between the states and is quantified by the Holevo-Helstrom bound~\cite{H73,H69}.

\begin{proposition}[Holevo-Helstrom bound~\cite{H73,H69}]\label{prop: holevo-helstrom}
    Let $\rho,\sigma\in\mathcal D(\mathcal H)$ be two quantum states and let \(p_\rho, p_\sigma\) be their prior probabilities, with $p_\sigma+p_\rho=1$. Then
    \begin{equation*}
        p_{\mathrm{succ}}(\rho,\sigma) := \max_{\mathbb I\geq\Lambda\geq 0} \left[ p_\rho \operatorname{Tr}(\Lambda\rho) + p_\sigma \operatorname{Tr}((\mathbb{I}-\Lambda)\sigma )\right] =\frac{1}{2}\left(1+\|p_\rho\rho-p_\sigma\sigma\|_1\right).
    \end{equation*}
\end{proposition}

In this work, algorithms are often restricted to perform polynomial-time operations.
In particular, we consider uniform/nonuniform algorithms that run in QPT, meaning that they run on a quantum Turing machine (in our case, implemented by quantum circuits) whose number of operations is bounded by some polynomial function of the length of its input.
We refer to the class of all positive polynomial functions on their input as $\operatorname{Poly}$.
If $\mathcal A$ is a (possibly randomized) algorithm, we denote by $y \leftarrow\mathcal{A}(x)$ the experiment of running the algorithm \(\mathcal{A}\) on input \(x\), outputting the value \(y\).

\begin{definition}[Negligible function]\label{def: negligible function}
A function $f:\mathbb N\rightarrow\mathbb R$ is said to be \emph{negligible} if 
\begin{equation*}
    \forall p\in\mathrm{Poly}\quad \exists\lambda_0\in\mathbb N \quad\text{such that}\quad \forall\lambda\geq\lambda_0 \quad |f(\lambda)| < \frac{1}{|p(\lambda)|}.
\end{equation*}
We denote the class of such functions as $\operatorname{Negl}$.
\end{definition}

\begin{definition}[Computational indistinguishability]\label{def: comp indist}
Let $\lambda\in\mathbb N$ be the security parameter. 
We say that two families of states $\{\rho^\lambda\}_\lambda$ and $\{\sigma^\lambda\}_\lambda$ are \emph{computationally indistinguishable}, denoted $\{\rho^\lambda\}_\lambda \approx\{\sigma^\lambda\}_\lambda$, if for all non-uniform QPT algorithms $\mathcal D$, there exists a negligible function \(\varepsilon\in\operatorname{Negl}\), such that, for all $\lambda\in\mathbb N$ and all polynomial-size advice states $\alpha_\lambda$, the advantage
\begin{equation*}
    \operatorname{Adv}_{\mathcal D^{\alpha^\lambda}}(\rho^\lambda,\sigma^\lambda) := \left|\operatorname{Pr}[1\leftarrow\mathcal{D}(\rho^\lambda,\alpha^\lambda)] - \operatorname{Pr}[1\leftarrow\mathcal{D}(\sigma^\lambda,\alpha^\lambda)]\right| \leq\varepsilon(\lambda).
\end{equation*}
\end{definition}

\medskip

\begin{remark}[Security parameter]
In cryptography, results are taken in the regime in which the argument of a negligible function, called the \emph{security parameter}, usually denoted by \(\lambda\), goes asymptotically to infinity (according to the definition of negligible function).
In this work, all results hold in this regime whenever a security parameter is considered. 
When clear from the context, this is implicit in the proofs to avoid notational overload.
\end{remark}

\subsection{Entanglement Measures}
We recall some relevant measures of entanglement of quantum states in both the information-theoretic and computational settings.
An entanglement measure is a scalar function defined on $\mathcal{D}(\mathcal{H})$, which does not increase under the action of LOCC channels~\cite{VPRK97}. 
First, we recall the definition of a LOCC channel (Definition~\ref{def:locc}) and the definitions of one-shot distillable entanglement and one-shot entanglement cost (Definitions~\ref{def:Ed} and~\ref{def:Ec}), then state their relation in Proposition~\ref{Th: ent distil and cost}. We then recall the definition of a vanishing function (Definition~\ref{def: vanishing function}) and remark on the behavior of the relation in Proposition~\ref{Th: ent distil and cost} when the error of distillation and dilution are both vanishing functions.
The distillable entanglement~\cite{BD10} is the largest number of EPR pairs one can obtain from a given bipartite quantum state by using only LOCC channels.
The entanglement cost~\cite{BD11} represents its ``converse'' quantity, which gives the least number of EPR pairs one needs to start with to reconstruct a given state using LOCC channels exclusively.

\begin{definition}[LOCC channel~\cite{KW24}]\label{def:locc}
An \emph{LOCC channel} is a map implementable by applying corresponding one-way LOCC channels a finite number of times. In turn, a \emph{one-way LOCC channel from Alice to Bob} $\Gamma^{A\rightarrow B}:\mathcal D(\mathcal H_A\otimes\mathcal H_B)\rightarrow\mathcal{D}(\mathcal H_{\bar{A}}\otimes\mathcal H_{\bar{B}})$ is a map of the form 
\begin{equation*}
    \Gamma^{A\rightarrow B}:=\sum_{x\in X}\mathcal{E}_{A}^x\otimes\mathcal{N}_{B}^x, 
\end{equation*}
from the input spaces $\mathcal H_A,\mathcal H_B$ to the output spaces $\mathcal H_{\bar{A}},\mathcal H_{\bar B}$, for a finite $|X|$, where the maps $\{\mathcal{E}^x_A\}_x$, $\{\mathcal N_B^x\}_x$ are completely positive, with each $\mathcal N_B^x$ and $\sum_{x\in X}\mathcal E_A^x$ being trace preserving. The corresponding one-way LOCC channel from Bob to Alice is defined similarly as
\begin{equation*}
    \Gamma^{B\rightarrow A}:=\sum_{x\in X}\mathcal{N}_{A}^x\otimes\mathcal{E}_{B}^x.
\end{equation*}
Then, an \emph{LOCC channel} $\Gamma: \mathcal D(\mathcal H_A\otimes\mathcal H_B)\rightarrow\mathcal{D}(\mathcal H_{\bar{A}}\otimes\mathcal H_{\bar{B}})$ can be written in the form  
\begin{equation*}
    \Gamma=\sum_{y\in Y}\mathcal{T}_{A}^y\otimes\mathcal{R}_{B}^y,
\end{equation*}
 for a finite $|Y|$, where the $\operatorname{maps}\left\{\mathcal{T}_A^y\right\}_y,\left\{\mathcal{R}_B^y\right\}_y$ are completely positive and $\Gamma$ is trace preserving.
\end{definition}

\begin{definition}[One-shot distillable entanglement \(E_D^\varepsilon\)~\cite{BD10}]\label{def:Ed}
Let $\rho_{AB} \in \mathcal{D}(\mathcal H_A\otimes \mathcal H_B)$ be a bipartite quantum state. 
The \emph{one-shot distillable entanglement}, $E_D^\varepsilon$, of a state $\rho_{AB}$, for $\varepsilon\in[0,1)$, is defined as
\begin{equation*}
   E_D^\varepsilon(\rho_{AB}) := \sup_{d,\Gamma}\left\{d:1-\operatorname{F}(\Gamma\left(\rho_{AB}),\Phi_{AB}^{\otimes d}\right)\leq \varepsilon\right\},
\end{equation*}
where the optimization is performed with regard to LOCC channels $\Gamma$ and the number of EPR pairs $d\in\mathbb N$.
\end{definition}

\begin{definition}[One-shot entanglement cost \(E_C^\varepsilon\)~\cite{BD11}]\label{def:Ec}
Let $\rho_{AB} \in \mathcal{D}(\mathcal H_A\otimes \mathcal H_B)$ be a bipartite quantum state. 
The \emph{one-shot entanglement cost}, $E_C^\varepsilon$, of a state $\rho_{AB}$, for $\varepsilon\in[0,1)$, is defined as
\begin{equation*}
    E_C^\varepsilon(\rho_{AB}) := \inf_{c,\Gamma}\{c:1-\operatorname{F}(\rho_{AB},\Gamma(\Phi_{AB}^{\otimes c}))\leq \varepsilon\},
\end{equation*}
where the optimization is performed with regard to LOCC channel $\Gamma$ and number of EPR pairs $c\in\mathbb N$.
\end{definition}
\begin{definition}[Vanishing at infinity function]\label{def: vanishing function}
    A function $g:\mathbb N\rightarrow \mathbb R$ is said to be \emph{vanishing at infinity} if 
    \begin{equation*}
        \forall\varepsilon>0\quad \exists x_0\in\mathbb N\quad\text{such that}\quad\forall x\geq x_0\quad |g(x)|\leq \varepsilon.
    \end{equation*}
\end{definition}
\begin{proposition}[Relation between $E_D^\varepsilon$ and $E_C^\varepsilon$~{\cite{Wilde21}}]\label{Th: ent distil and cost}
    Let $\rho_{AB}\in\mathcal D(\mathcal H_A\otimes \mathcal H_B)$ be a bipartite quantum state, and let $\varepsilon_1, \varepsilon_2\in[0,1)$ be the errors.
    Then,
    \begin{equation*}
        E_D^{\varepsilon_1} (\rho_{AB})\leq E_C^{\varepsilon_2}(\rho_{AB}) + \operatorname{log}_2\left(\frac{1}{1-\varepsilon^\prime}\right),
    \end{equation*}
    where \( \varepsilon^\prime:=\left(\sqrt{\varepsilon_1} + \sqrt{\varepsilon_2}\right)^2<1\).
\end{proposition}
\begin{remark}\label{remark: relation for eps negl}
    We will frequently manipulate states such that \(\varepsilon_1,\varepsilon_2\) go to zero as their inputs grow, i.e., are vanishing at infinity in some parameter \(x\).
    In that case, when $\varepsilon_1$ and $\varepsilon_2$ are both equal to some function $\varepsilon$ that is vanishing at infinity, the relation in Proposition~\ref{Th: ent distil and cost} can be stated, for all sufficiently large \(x>x_0\) for some \(x_0\in\mathbb{N}\), as
    \(
        E^{\varepsilon}_D(\rho_{AB}) \leq  E_C^\varepsilon(\rho_{AB}).
    \)
    Since $E^{\varepsilon}_D(\rho_{AB})$ and $E^{\varepsilon}_C(\rho_{AB})$ (Definitions~\ref{def:Ed} and~\ref{def:Ec}) are integer-valued functions and the logarithm term is strictly smaller than 1, taking the floor is sufficient to show this.
\end{remark}

The entanglement measures from Definitions~\ref{def:Ed} and~\ref{def:Ec} consider arbitrary, possibly inefficient, operations that two parties may perform in an LOCC setting.
In~\cite{ABV23}, a new framework is introduced to characterize an operational setting in which parties are restricted to perform efficient LOCC operations, and the computational entanglement of bipartite quantum states is studied therein.
Below, we recall the definitions of one-shot computational distillable entanglement and one-shot computational entanglement cost.
Each definition covers the case in which for each size \(\lambda\in \mathbb N\) there is one state (Definitions~\ref{def: comp dist} and~\ref{def: comp cost}), and a uniform version where for each size \(\lambda\in \mathbb N\) there are \(2^{\kappa(\lambda)}\) states, for some \(\kappa\in\operatorname{Poly}\) (Definitions~\ref{def: uni comp dist} and~\ref{def: uni comp cost}).

Following~\cite{ABV23}, we work with bipartite qubit quantum systems.
We denote the qubit space of Alice as $\mathcal H_A:={(\mathbb C^2)}^{\otimes n_A}$ and the qubit space of Bob as $\mathcal H_B:={(\mathbb C^2)}^{\otimes n_B}$, where $n_A,n_B$ are the number of qubits of Alice and Bob, respectively. We refer to quantum states in joint Hilbert spaces of Alice and Bob as $n(\lambda) =n_A(\lambda)+n_B(\lambda)$-qubit bipartite quantum states.

In this work, the number of qubits of Alice and Bob are polynomial functions of a parameter \(\lambda\in\mathbb{N}\), i.e., \(n_A,n_B: \mathbb{N}\to \mathbb{N}\) with \(n_A,n_B \in \operatorname{Poly}\).
Hence, we denote corresponding spaces of Alice and Bob as $\mathcal{H}_A^\lambda$ and $\mathcal{H}_B^\lambda$.
When considering a property related to security, it is considered with respect to the $\lambda$ parameter (the security parameter), and the results hold for $\lambda$ large enough, i.e., for all \(\lambda > \lambda_0\) for some \(\lambda_0\in\mathbb N\).

\begin{definition}[Circuit representation of LOCC channel~\cite{ABV23}] \label{def:LOCC circ}
    Let $\Gamma$ be an LOCC channel with $r$ denoting the number of rounds such that
    \begin{equation*}
    \Gamma:\mathcal D\left(\mathcal H_A^\lambda\otimes \mathcal H_B^\lambda\right) \longrightarrow\mathcal D\left(\mathcal H^\lambda_{\bar A} \otimes \mathcal H^\lambda_{\bar B}\right),
    \end{equation*}
    where $n_A, n_B$ and $\bar n_A, \bar n_B$ are, respectively, the numbers of input and output qubits of Alice and Bob. We say that $\Gamma$ admits a \emph{circuit representation} if there exist two sets of quantum circuits $\{\mathcal C_A^i\}_{i=1}^r,\{\mathcal C_B^i\}_{i=1}^r$ acting on registers $A,A^\prime, C$ and $B,B^\prime, C$, respectively. $A^\prime,B^\prime$ represent ancilla qubits and $C$ is a shared communication register. 

    Alice and Bob initialize their ancilla qubits to $|0\rangle$ and the shared communication register to $|0\rangle$. At each round $j\in[r]$, Alice applies $\mathcal C_A^j$ on $(A,A^\prime,C)$, measures $C$ in the computational basis, and sends the result to Bob. Based on the outcome of the measurement, Bob applies the circuit $\mathcal C_B^j$ on $(B,B^\prime, C)$, measures $C$ and sends the result back to Alice.
    
    The output of the channel $\Gamma$ is considered to be the final state stored in the registers $\bar A, \bar B$.
\end{definition}
\begin{definition}[Efficient family of LOCC channels~\cite{ABV23}]\label{def: eff locc}
    We say that the family of LOCC channels $\{\hat\Gamma^\lambda\}_{\lambda\in\mathbb{N}}$ is \emph{efficient} if the resources needed to implement the circuit representation of $\hat\Gamma^\lambda$ are upper-bounded by some polynomial function. Formally,
    \begin{equation*}
       \exists \,p \in\operatorname{Poly}\quad \text{such that}\quad \forall \,\lambda\in\mathbb N,\quad t(\lambda) \leq  p(\lambda),
    \end{equation*}
    where $t(\lambda)$ is the total number of gates, measurements in the computational basis, and ancillas needed to realize the circuit representation of $\hat\Gamma^\lambda$.
\end{definition}

\begin{definition}[One-shot computational distillable entanglement \(\hat{E}_D^\varepsilon\)~\cite{ABV23}]\label{def: comp dist}
Let \(\lambda\in\mathbb{N}\), $\varepsilon:\mathbb N\rightarrow[0,1]$ be an arbitrary function, and $n_A,n_B:\mathbb N\rightarrow\mathbb N$ be arbitrary polynomially bounded functions.
Let $\{\rho^\lambda_{AB}\}_\lambda$ be a family of \(n_A(\lambda) + n_B(\lambda)\)-qubit bipartite states with \(\rho^\lambda_{AB}\in\mathcal D\left(\mathcal H_A^\lambda \otimes \mathcal H_B^\lambda\right)\). \\
The \emph{one-shot computational distillable entanglement} of the family $\{\rho^\lambda_{AB}\}_\lambda$ has a valid lower bound $m:\mathbb{N}\to\mathbb{N}$, denoted
\begin{equation*}
    \hat E_D^\varepsilon(\{\rho_{AB}^\lambda\}_\lambda)\geq m,
\end{equation*}
if there exists an efficient family of LOCC channels, $\{\hat \Gamma^\lambda\}_\lambda$, such that, for all \(\lambda\),
\begin{equation*}
    p_{\mathrm{err}}(\hat\Gamma^\lambda,\rho^\lambda_{AB}):= 1 - \operatorname{F}(\hat\Gamma^\lambda(\rho^\lambda_{AB}), \Phi^{\otimes m(\lambda)}_{AB}))\leq \varepsilon(\lambda).
\end{equation*}
\end{definition}

\begin{definition}[One-shot computational entanglement cost \(\hat{E}_C^\varepsilon\)~\cite{ABV23}]\label{def: comp cost}
Let \(\lambda\in\mathbb{N}\), $\varepsilon:\mathbb N\rightarrow[0,1]$ be an arbitrary function, and  $n_A,n_B:\mathbb N\rightarrow\mathbb N$ be arbitrary polynomially bounded functions.
Let $\{\rho^\lambda_{AB}\}_\lambda$ be a family of \(n_A(\lambda) + n_B(\lambda)\)-qubit bipartite states with \(\rho^\lambda_{AB}\in\mathcal D\left(\mathcal H_A^\lambda \otimes \mathcal H_B^\lambda\right)\).\\
The \emph{one-shot computational entanglement cost} of $\{\rho^\lambda_{AB}\}_\lambda$ has a valid upper bound $n:\mathbb{N}\to\mathbb{N}$, denoted
\begin{equation*}
    \hat E_C^\varepsilon(\{\rho_{AB}^\lambda\}_\lambda)\leq n,
\end{equation*}
if there exists an efficient family of LOCC channels $\{\hat \Gamma^\lambda\}_\lambda$, such that, for all \(\lambda\),
\begin{equation*}
    p_{\mathrm{err}}(\hat\Gamma^\lambda,\rho^\lambda_{AB}):= 1 - \operatorname{F}(\hat\Gamma^\lambda(\Phi^{\otimes n(\lambda)}_{AB}), \rho^\lambda_{AB}))\leq \varepsilon(\lambda).
\end{equation*}
\end{definition}
    
\begin{definition}[Uniform one-shot computational distillable entanglement~\cite{ABV23}]\label{def: uni comp dist}
    Let \(\lambda\in\mathbb{N}\), $\varepsilon:\mathbb N\rightarrow[0,1]$ be an arbitrary function, and $n_A,n_B,\kappa:\mathbb N\rightarrow\mathbb N$ be arbitrary polynomially bounded functions.
    Let $\{\rho^k_{AB}\}_{k\in\{0,1\}^{\kappa(\lambda)}}$ be a family of \(n_A(\lambda) + n_B(\lambda)\)-qubit bipartite states $\rho^k_{AB}\in\mathcal D\left(\mathcal H_A^\lambda \otimes \mathcal H_B^\lambda\right)$.\\
    The \emph{uniform one-shot computational distillable entanglement} of  $\{k,\rho^k_{AB}\}_{k\in\{0,1\}^{\kappa(\lambda)}}$ has a valid lower bound \(m:\mathbb{N}\to\mathbb{N}\), denoted 
    \begin{equation*}
        \hat E_D^\varepsilon(\{k,\rho_{AB}^k\}_{k})\geq m,
    \end{equation*}
    if there exists an efficient family of LOCC channels $\{\hat \Gamma^\lambda\}_\lambda$, such that, for all \(\lambda\in \mathbb{N}\) and all \(k\in \{0,1\}^{\kappa(\lambda)}\),
    \begin{equation*}
         p_{\mathrm{err}}(\hat\Gamma^\lambda,k,\rho^k_{AB}):= 1 - \operatorname{F}(\hat\Gamma^\lambda(k,\rho^k_{AB}), \Phi^{\otimes m(\lambda)}_{AB}))\leq \varepsilon(\lambda).    
    \end{equation*}    
\end{definition}

\begin{definition}[Uniform one-shot computational entanglement cost~\cite{ABV23}]\label{def: uni comp cost}
   Let \(\lambda\in\mathbb{N}\), $\varepsilon:\mathbb N\rightarrow[0,1]$ be an arbitrary function, and $n_A,n_B,\kappa:\mathbb N\rightarrow\mathbb N$ be arbitrary polynomially bounded functions.
    Let $\{\rho^k_{AB}\}_{k\in\{0,1\}^{\kappa(\lambda)}}$ be a family of \(n_A(\lambda) + n_B(\lambda)\)-qubit bipartite states $\rho^k_{AB}\in\mathcal D\left(\mathcal H_A^\lambda \otimes \mathcal H_B^\lambda\right)$.\\
    The \emph{uniform one-shot computational entanglement cost} of  $\{k,\rho^k_{AB}\}_{k\in\{0,1\}^{\kappa(\lambda)}}$ has a valid upper bound \(n:\mathbb{N}\to\mathbb{N}\), denoted 
    \begin{equation*}
        \hat E_C^\varepsilon(\{k,\rho_{AB}^k\}_{k})\leq n,
    \end{equation*}
    if there exists an efficient family of LOCC channels $\{\hat \Gamma^\lambda\}_\lambda$, such that, for all \(\lambda\in \mathbb{N}\) and all \(k\in \{0,1\}^{\kappa(\lambda)}\),
    \begin{equation*}
        p_{\mathrm{err}}(\hat\Gamma^\lambda,k,\rho^k_{AB}):= 1 - \operatorname{F}(\hat\Gamma^\lambda(k,\Phi^{\otimes n(\lambda)}_{AB}), \rho^k_{AB}))\leq \varepsilon(\lambda).
    \end{equation*}    
\end{definition}

\smallskip
\noindent\textit{Notation.}
    In Definitions~\ref{def: uni comp dist} and~\ref{def: uni comp cost}, the notation $\hat \Gamma^\lambda(k,\rho^k_{AB})$ stands for $\hat \Gamma^\lambda$ evaluated on the state $\ketbra{k}{k}_{\tilde A}\otimes \rho_{AB}^k \otimes \ketbra{k}{k}_{\tilde B}$ (Definition~\ref{def: uni comp dist})  and $\ketbra{k}{k}_{\tilde A}\otimes \Phi_{AB}^{\otimes n(\lambda)} \otimes \ketbra{k}{k}_{\tilde B}$ (Definition~\ref{def: uni comp cost}), where $\tilde A, \tilde B$ are the key registers of Alice and Bob, respectively. We also make use of a shorthand notation for a keyed family of states $\{k,\rho^k_{AB}\}_{k\in\{0,1\}^{\kappa(\lambda)}}$ and write $\{k,\rho^k_{AB}\}_k$.

\smallskip
\begin{remark}
    For generality of our results, in this work, we use both the concepts of negligible functions (Definition~\ref{def: negligible function}) and vanishing at infinity functions (Definition~\ref{def: vanishing function}). 
    In particular, we utilize the notion of negligible functions when addressing computational indistinguishability, while the concept of vanishing at infinity is used for the results regarding the error values of entanglement measures and trace distances. 
    Note that any negligible function constitutes a vanishing at infinity function, making the latter a weaker assumption.
\end{remark}

\subsection{Pseudoentanglement}
In this section, we recall different definitions of pseudoentanglement for general bipartite mixed states. 
They aim to represent the same idea: a family of states $\{k,\psi^k\}_k$ is said to be pseudoentangled if there exists another computationally indistinguishable reference family $\{k,\phi^k\}_k$ that possesses more entanglement.  

We work with and study the relationship between cryptography and the operational definitions of pseudoentanglement introduced in~\cite{ABV23} and adapted in~\cite{GE24}. 
These two definitions have two main differences.
The first one is that in~\cite{ABV23} the entanglement of the states in the reference family $\{k,\phi^k\}_k$ must be efficiently distillable. In contrast, in~\cite{GE24} it is only considered that the entanglement in the reference family exists, i.e., is accessible to unbounded parties (motivated by the reference family $\{k,\phi^k\}_k$ not being required to be efficiently preparable). The second main difference is that, in~\cite{GE24}, neither family of states has keys (corresponding to the situation where the families consist of a single state for each size parameter \(\lambda\)).
We refer to the original definition~\cite{ABV23} as \emph{fully-computational pseudoentanglement} and to the definition of~\cite{GE24} as \emph{inefficiently-distillable pseudoentanglement}.

\begin{definition}[Fully-computational pseudoentanglement~\cite{ABV23}]\label{def: pseudo-entanglement ABV}
     Let \(\lambda\in\mathbb{N}\) be the security parameter, $n,\kappa:\mathbb N\rightarrow\mathbb N$ be arbitrary polynomially bounded functions, and \(\varepsilon:\mathbb{N}\to [0,1]\) with \(c,d:\mathbb{N}\to\mathbb{N}\) be arbitrary functions.\\
     A family of \(n(\lambda)\)-qubit bipartite states \(\{k,\psi_{AB}^k\}_{k\in{\{0,1\}^{\kappa(\lambda)}}}\) is said to be \emph{\((\varepsilon,c,d)\)-pseudoentangled} if there exists a family of \(n(\lambda)\)-qubit bipartite states  \(\{k,\phi_{AB}^k\}_{k\in{\{0,1\}^{\kappa(\lambda)}}}\) such that
     \begin{enumerate}[(i)]
         \item The uniform computational one-shot entanglement cost of the family $\{k,\psi^k_{AB}\}_{k}$ is upper-bounded by $c$, i.e., $\hat E_C^\varepsilon(\{k,\psi_{AB}^k\}_k)\leq c$.
         \item The uniform computational one-shot distillable entanglement of the family  $\{k,\phi_{AB}^k\}_{k}$ is lower-bounded by $d$, i.e., $\hat E_D^\varepsilon( \{k,\phi_{AB}^k\}_k)\geq d$.
         \item For all polynomial functions \(q:\mathbb{N}\to\mathbb{N}\), 
         $\left\{\overline{\psi^{\otimes q(\lambda)}_{{A}{B}}}^\lambda \right\}_\lambda\approx \left\{  \overline{\phi^{\otimes q(\lambda)}_{ A  B}}^\lambda\right\}_\lambda$, where  $\overline{\psi^{\otimes q(\lambda)}_{{A}{B}}}^\lambda := \mathbb E_k\left[  \psi^{k\,\otimes q(\lambda)}_{AB}\right]$ and $\overline{\phi^{\otimes q(\lambda)}_{ A  B}}^\lambda := \mathbb E_k\left[  \phi^{k\,\otimes q(\lambda)}_{AB}\right]$.
     \end{enumerate}
\end{definition}
The other relevant definition of pseudoentanglement was introduced in~\cite{GE24}.
Instead of requiring the states with more entanglement to be computationally distillable, as in Definition~\ref{def: pseudo-entanglement ABV} \textit{(ii)}, it imposes a requirement on the one-shot information-theoretic distillable entanglement instead (Definition~\ref{def: pseudo-entanglement GE} \textit{(ii)}).

\begin{definition}[Inefficiently-distillable pseudoentanglement~\cite{GE24}]\label{def: pseudo-entanglement GE}
     Let \(\lambda\in\mathbb{N}\) be the security parameter, $n:\mathbb N\rightarrow\mathbb N$ be an arbitrary polynomially bounded function, and \(\varepsilon:\mathbb{N}\to [0,1]\) with \(c,d:\mathbb{N}\to\mathbb{N}\) be arbitrary functions.\\
     A family of \(n(\lambda)\)-qubit bipartite states \(\{\psi_{AB}^\lambda\}_{\lambda}\) is said to be \emph{\((\varepsilon,c,d)\)-pseudoentangled} if there exists a family of \(n(\lambda)\)-qubit bipartite states \(\{\phi_{AB}^\lambda\}_{\lambda}\) such that
     \begin{enumerate}[(i)]
         \item The computational one-shot entanglement cost of the family $\{\psi^\lambda_{AB}\}_{\lambda}$ is upper-bounded by $c$, i.e., $\hat E_C^\varepsilon(\{\psi_{AB}^\lambda\})\leq c$.
         \item The one-shot distillable entanglement of the family  $\{\phi_{AB}^\lambda\}_{\lambda}$ is lower-bounded by $d$, i.e., $E_D^\varepsilon( \phi_{AB}^\lambda)\geq d$, for all \(\lambda\in \mathbb N\).
         \item For all polynomial functions \(q:\mathbb{N}\to\mathbb{N}\), $\left\{{\psi^{\lambda\,\otimes q(\lambda)}_{AB}} \right\}_\lambda\approx \left\{{\phi^{\lambda\,\otimes q(\lambda)}_{AB}}\right\}_\lambda$.
     \end{enumerate}
\end{definition}
As noted in~\cite[Section 5]{ABV23}, the concept of pseudoentanglement, and hence its definitions, is only interesting when there is a gap in entanglement between the reference entangled and the pseudoentangled families of states, i.e., $c(\lambda)<d(\lambda)$.
We further assume that both families in the pseudoentanglement structure are preparable in polynomial time.
While this is an additional assumption, not explicitly required in the original definitions of pseudoentanglement, it is a natural requirement when one wants to leverage this resource in operational scenarios, as customary in the current constructions in the state-of-the-art~\cite{ABV23,GE24,ABF24,BF24}.

\subsection{EFI Pairs}
In this section, we recall the concept of EFI pairs~\cite{BCQ23}.
Intuitively, EFI pairs are two families of quantum states that are efficiently preparable by uniform QPT algorithms, information-theoretically distinguishable (have large trace distance), but indistinguishable to any quantum algorithm that is restricted to run in QPT.

\begin{definition}[EFI pairs~\cite{BCQ23}]\label{def: EFI pair}
    Let \(\lambda\in\mathbb{N}\) be the security parameter.
    The pair of states $(\rho_0^\lambda, \rho_1^\lambda)_\lambda$ is said to be an \emph{EFI pair} if the following conditions hold:
     \begin{enumerate}[(i)]
        \item \emph{Efficient preparation:} There exists a uniform QPT algorithm $\mathcal A$ that on input \((1^\lambda, b)\) outputs $\rho^\lambda_b$, for $b\in\{0,1\}$.
        \item \emph{Noticeable statistical distance:} There exists a positive polynomial $p \in\operatorname{Poly}$, such that, for \(\lambda\) large enough, $\operatorname{TD}(\rho^\lambda_0,\rho^\lambda_1)   \geq \frac{1}{p(\lambda)}$.
        \item \emph{Computational indistinguishability:} \(\{\rho_0^\lambda\}_\lambda \approx \{\rho_1^\lambda\}_\lambda\).
    \end{enumerate}
\end{definition}

\section{Error-Continuity of Computational Entanglement Cost and Gap in Computational Entanglement}\label{sec: error continuity and gap}
In this section, we introduce and show a continuity relation connecting the trace distance of two families to their computational entanglement cost, i.e., the ability to prepare them efficiently from EPR pairs employing efficient LOCC channels.
Furthermore, this allows us to demonstrate that families of states that are statistically close (i.e., have a vanishing trace distance) do not exhibit a gap between their respective computational distillable entanglement and cost.

We start by showing this continuity relation for the one-shot computational entanglement cost $\hat E_C^\varepsilon$ when starting from the same number of EPR pairs, $n(\lambda)$.
\begin{lemma}[Error-continuity of $\hat E_C^\varepsilon$]\label{lemma: continuity cost to the same epr}
    Let \(\lambda\in\mathbb{N}\), \(m:\mathbb{N}\to\mathbb{N}\) be a polynomially bounded function and \(n:\mathbb{N}\to\mathbb{N}\) an arbitrary function.
    Let $\{\rho_{AB}^\lambda\}_\lambda$, $\{\sigma_{AB}^\lambda\}_\lambda$ be two families of $m(\lambda)$-qubit bipartite states.
    If, for all \(\lambda\), $\frac{1}{2}\|\rho_{AB}^\lambda - \sigma_{AB}^\lambda\|_1\leq\delta(\lambda)$, with $ \delta(\lambda)+\sqrt{\varepsilon(\lambda)}\leq 1$, and $\hat E_C^\varepsilon(\{\rho_{AB}^\lambda\}_\lambda)\leq n$, then,
    \begin{equation*}
        \hat E_C^{\varepsilon^\prime}(\{\sigma_{AB}^\lambda\}_\lambda) \leq n,
    \end{equation*}
    where $\varepsilon^\prime(\lambda) = \left(\delta(\lambda)+\sqrt{\varepsilon(\lambda)}\right)\left(2-\delta(\lambda)-\sqrt{\varepsilon(\lambda)}\right)$.
\end{lemma}
\begin{proof}
        Consider the assumption that $\hat E_C^\varepsilon(\{\rho_{AB}^\lambda\}_\lambda)\leq n$.
    Then, there exists an efficient family of LOCC channels \(\{\hat\Gamma^\lambda\}_{\lambda}\), such that, for all \(\lambda\),
    \begin{equation}\label{eq:err-cont-1}
        p_{\text{err}}(\hat\Gamma^\lambda,\rho_{AB}^\lambda) = 1-\operatorname{F} \left(\rho_{AB}^\lambda, \hat\Gamma^\lambda\left(\Phi_{AB}^{\otimes n(\lambda)}\right)\right) \leq \varepsilon(\lambda).
    \end{equation}
    From the Fuchs-van de Graaf inequality (Proposition~\ref{prop: fvdg}),
\begin{equation}\label{eq: 1-dist btw rho and epr}
    \frac{1}{2}\left\|\rho_{AB}^\lambda - \hat\Gamma^\lambda\left(\Phi_{AB}^{\otimes n(\lambda)}\right)\right\|_1\leq\sqrt{1-\operatorname{F}\left(\rho_{AB}^\lambda, \hat\Gamma^\lambda\left(\Phi_{AB}^{\otimes n(\lambda)}\right)\right)}\leq\sqrt{\varepsilon(\lambda)}.
\end{equation}
Also, by the triangle inequality,
    \begin{equation}\label{eq: triangle dilute}
        \frac{1}{2}\left\|\hat\Gamma^\lambda\left(\Phi_{AB}^{\otimes n(\lambda)}\right) - \sigma_{AB}^\lambda\right\|_1 \leq \frac{1}{2}\left\|\hat\Gamma^\lambda\left(\Phi_{AB}^{\otimes n(\lambda)}\right) - \rho_{AB}^\lambda\right\|_1 + \frac{1}{2}\Big\|\rho_{AB}^\lambda-\sigma_{AB}^\lambda\Big\|_1.
    \end{equation}
Hence, combining Equations~\eqref{eq: 1-dist btw rho and epr} and~\eqref{eq: triangle dilute},
\begin{equation}
    \frac{1}{2}\left\|\sigma_{AB}^\lambda - \hat\Gamma^\lambda\left(\Phi_{AB}^{\otimes n(\lambda)}\right)\right\|_1 \leq \delta(\lambda) + \sqrt{\varepsilon(\lambda)},
\end{equation}
since by assumption $\frac{1}{2}\|\rho_{AB}^\lambda - \sigma_{AB}^\lambda\|_1\leq\delta(\lambda)$.

Applying the Fuchs-van de Graaf inequality (Proposition~\ref{prop: fvdg}) again gives
\begin{align}
    1-\sqrt{\operatorname{F}\left(\sigma_{AB}^\lambda, \hat\Gamma^\lambda\left(\Phi_{AB}^{\otimes n(\lambda)}\right)\right)} \leq\frac{1}{2}\left\|\sigma_{AB}^\lambda - \hat\Gamma^\lambda\left(\Phi_{AB}^{\otimes n(\lambda)}\right)\right\|_1 \leq \delta(\lambda) + \sqrt{\varepsilon(\lambda)},
\end{align}
thus,
\begin{align}
    \operatorname{F}\left(\sigma_{AB}^\lambda, \hat\Gamma^\lambda\left(\Phi_{AB}^{\otimes n(\lambda)}\right)\right) \geq\left(1-\delta(\lambda)-\sqrt{\varepsilon(\lambda)}\right)^2. \label{eq: bound on fidelity}
\end{align}

Expanding the definition of the error of dilution, \(p_\mathrm{err}\), for the family $\{\sigma_{AB}^\lambda\}_\lambda$, using LOCC channels $\{\hat\Gamma^\lambda\}_\lambda$, and plugging Equation~\eqref{eq: bound on fidelity} back into~\eqref{eq:err-cont-1} gives
\begin{align}\begin{split}
     p_{\text{err}}(\hat\Gamma^\lambda, \sigma_{AB}^\lambda) 
     &\leq 1-\left(1-\delta(\lambda)-\sqrt{\varepsilon(\lambda)}\right)^2 \\
     &= \left(\delta(\lambda)+\sqrt{\varepsilon(\lambda)}\right)\left(2-\delta(\lambda)-\sqrt{\varepsilon(\lambda)}\right),
\end{split}\end{align}
which establishes the claim of the lemma.
\end{proof}

We remark that the same derivation holds if one substitutes the definition of one-shot computational cost for its uniform counterpart, since it quantifies over all keys \(k\).
Therefore, we get the following corollary.
\begin{corollary}\label{cor: continuity cost uniform}
    Let \(\lambda\in\mathbb{N}\), $m,\kappa:\mathbb N\rightarrow\mathbb N$ be polynomially bounded functions and $n:\mathbb{N}\to\mathbb{N}$ be an arbitrary function.
    Let $\{k,\rho_{AB}^k\}_{k\in \{0,1\}^{\kappa(\lambda)}}$ and $\{k,\sigma_{AB}^k\}_{k\in \{0,1\}^{\kappa(\lambda)}}$ be two families of $m(\lambda)$-qubit bipartite states.
    If, for all \(k,\lambda\), $\frac{1}{2}\|\rho_{AB}^k - \sigma_{AB}^k\|_1\leq\delta(\lambda)$, with $\delta(\lambda)+\sqrt{\varepsilon(\lambda)}\leq 1$, and $\hat E_C^\varepsilon(\{k,\rho_{AB}^k\}_k)\leq n$, then,
    \begin{equation*}
        \hat E_C^{\varepsilon^\prime}(\{k,\sigma_{AB}^k\}_k)\leq n,
    \end{equation*}
    where $\varepsilon^\prime(\lambda):= \left(\delta(\lambda)+\sqrt{\varepsilon(\lambda)}\right)\left(2-\delta(\lambda)-\sqrt{\varepsilon(\lambda)}\right)$.
\end{corollary}

With this continuity relation of $\hat E_C^\varepsilon$, we can now show the result that establishes the connection between the trace distance of two families of states and their gap in computational entanglement. Namely, we prove a statement relating the small trace distance between states in two families to a gap in one-shot computational entanglement cost and one-shot computational distillable entanglement of the corresponding families of states. 

\begin{lemma}[TD and entanglement gap]\label{lemma: connection of td and cd}
    Let \(\lambda \in\mathbb{N}\), $\varepsilon:\mathbb N\to [0,1]$ be a vanishing at infinity function, $c,d:\mathbb N\to \mathbb N$ be arbitrary functions, and \(n:\mathbb N\to\mathbb N\) be a polynomially bounded function.
     Let $\{\rho^\lambda_0\}_{\lambda}$ and $\{\rho_1^\lambda\}_{\lambda}$ be two families of $n(\lambda)$-qubit bipartite states, such that $\hat E_C^\varepsilon(\{\rho^\lambda_0\}_\lambda)\leq c$ and $\hat E_D^\varepsilon(\{\rho_1^\lambda\}_\lambda)\geq d$. Assume that $\frac{1}{2}\|\rho_1^\lambda-\rho_0^\lambda\|_1\leq\delta(\lambda)$ for all \(\lambda\), and $\lim_{\lambda\to\infty} \delta(\lambda)=  0$.
     Then,
    \begin{equation*}
        \exists \lambda_\star \in\mathbb N \quad\text{such that}\quad\forall\lambda\geq\lambda_\star,\quad c(\lambda)\geq d(\lambda).
    \end{equation*}
\end{lemma}
\begin{proof}
    To establish the claim of this lemma, we first relate the assumption on the computational one-shot distillable entanglement of the family $\{\rho_1^\lambda\}_\lambda$ to the requirement on the one-shot computational entanglement cost of the same family. 
    We then utilize the continuity statement for one-shot computational cost, as proved in Lemma~\ref{lemma: continuity cost to the same epr}, to obtain the desired result. 

    First, note, with a slight abuse of notation, that for any family of bipartite mixed states $\{\rho_{AB}^\lambda\}_\lambda$,
    \begin{equation}\label{eq:comp inequalities}
        \hat{E}^\varepsilon_C(\{\rho_{AB}^\lambda\}_\lambda)\geq E_C^\varepsilon(\{\rho_{AB}^\lambda\}_\lambda)\geq E^\varepsilon_D(\{\rho_{AB}^\lambda\}_\lambda)\geq\hat E^\varepsilon_D(\{\rho_{AB}^\lambda\}_\lambda),
    \end{equation}
    where we define the function \(E^\varepsilon_{D(C)}(\{\rho_{AB}^\lambda\}_\lambda) := E^\varepsilon_{D(C)}(\rho_{AB}^\lambda)\) pointwise for all \(\lambda\). 
    We prove every inequality in the chain, starting from the rightmost one.
    The first inequality, $E^\varepsilon_D(\{\rho_{AB}^\lambda\})\geq\hat E^\varepsilon_D(\{\rho_{AB}^\lambda\})$, holds due to a trivial inclusion of the class of efficient LOCC channels in the class of general LOCC channels. Therefore, for any valid lower bound $d$ on $\hat E^\varepsilon_D(\{\rho_{AB}^\lambda\})$ (Definition~\ref{def: comp dist}), for all values of $\lambda$, it holds that $E^\varepsilon_D(\rho_{AB}^\lambda)\geq d(\lambda)$~\cite{ABV23}. 
    The second inequality, $E_C^\varepsilon(\{\rho_{AB}^\lambda\})\geq E^\varepsilon_D(\{\rho_{AB}^\lambda\})$, holds for large enough $\lambda$ as a result of Proposition~\ref{Th: ent distil and cost} combined with the assumption of $\varepsilon(\lambda)$ being a vanishing at infinity function (see Remark~\ref{remark: relation for eps negl}). 
    The third inequality again follows from the inclusion of efficient LOCC channels in general LOCC channels, hence,  for any valid upper bound $c$ on $\hat E_C^\varepsilon(\{\rho_{AB}^\lambda\}_\lambda)$ (Definition~\ref{def: comp cost}), for all values of $\lambda$, $c(\lambda)\geq E^\varepsilon_C(\rho^\lambda_{AB})$~\cite{ABV23}.

    We now use the error-continuity of \(\hat E_C\) shown in  Lemma~\ref{lemma: continuity cost to the same epr}.
    By assumption, the trace distance between the states of the families $\{\rho^\lambda_0\}_\lambda$ and $\{\rho_1^\lambda\}_\lambda$ is upper-bounded as \(\frac{1}{2}\|\rho_0^\lambda-\rho_1^{\lambda}\|_1\leq\delta(\lambda)\) for  $\lambda\in\mathbb N$.
    Therefore, if $\hat{E}_C^\varepsilon(\{\rho_0^\lambda\}_\lambda)\leq c$, then
    \begin{equation}\label{eq:Ec 1}
         \hat{E}_C^{\varepsilon^\prime}(\{\rho_1^\lambda\}_\lambda)\leq c,
    \end{equation}
    where $\varepsilon^\prime (\lambda) = \left(\delta(\lambda)+\sqrt{\varepsilon(\lambda)}\right)\left(2-\sqrt{\varepsilon(\lambda)}-\delta(\lambda)\right)$.

    Finally, recall that, by assumption, $\hat E_D^\varepsilon(\{\rho_1^\lambda\}_\lambda)\geq d$. 
    Also, the assumption of $\varepsilon$ being vanishing at infinity implies that $\varepsilon^\prime(\lambda)$ is greater than $\varepsilon(\lambda)$ for large enough $\lambda$. Therefore, the valid lower bound $\hat E^{\varepsilon^\prime}_D(\{\rho_1^\lambda\})\geq d$ must hold. Then, from the chain of inequalities in Equation~\eqref{eq:comp inequalities},
    \begin{equation}\label{eq:Ec 2}
        \exists\lambda_\star\in\mathbb N\quad\text{such that}\quad\forall\lambda\geq\lambda_\star\quad c(\lambda)\geq d(\lambda),
    \end{equation}
    where $c$ is the valid upper bound on $\hat{E}_C^{\varepsilon^\prime}(\{\rho_1^\lambda\}_\lambda)$ in Equation~\eqref{eq:Ec 1}. This concludes the proof.
\end{proof}

Furthermore, the same result holds if the entanglement of the family \(\{\rho_1^\lambda\}_\lambda\) were quantified by the information-theoretic distillable entanglement instead of its computational version (motivated by the inefficiently-distillable pseudoentanglement --- Definition~\ref{def: pseudo-entanglement GE}).
This is stated in Corollary~\ref{cor: connection of td and cd_in}, below.
\begin{corollary}\label{cor: connection of td and cd_in}
      Let \(\lambda \in\mathbb{N}\), $\varepsilon:\mathbb N\to [0,1]$ be a vanishing at infinity function, $c,d:\mathbb N\to \mathbb N$ be arbitrary functions, and \(n:\mathbb N\to\mathbb N\) be a polynomially bounded function.
     Let $\{\rho^\lambda_0\}_{\lambda}$ and $\{\rho_1^\lambda\}_{\lambda}$ be two families of $n(\lambda)$-qubit bipartite states, such that $\hat E_C^\varepsilon(\{\rho^\lambda_0\}_\lambda)\leq c$ and $E_D^\varepsilon(\rho_1^\lambda)\geq d(\lambda)$ for all $\lambda$. Assume that $\frac{1}{2}\|\rho_1^\lambda-\rho_0^\lambda\|_1\leq\delta(\lambda)$ for all \(\lambda\) and $\lim_{\lambda\to\infty} \delta(\lambda)=  0$.
     Then,
    \begin{equation*}
        \exists \lambda_\star\in\mathbb N \quad\text{such that}\quad\forall\lambda\geq\lambda_\star,\quad c(\lambda)\geq d(\lambda).
    \end{equation*}
\end{corollary}
\begin{proof}
    The proof is the same as in Lemma~\ref{lemma: connection of td and cd}, but starting from the second-to-last inequality of Equation~\eqref{eq:comp inequalities} (ignoring the rightmost one).
\end{proof}

\section{Fully-Computational Pseudoentanglement is Sufficient for EFI Pairs}\label{sec: pseudo abv to efi}
Here, we prove the first of our two main theorems, which states that the existence of fully-computational pseudoentanglement~\cite{ABV23} (Definition~\ref{def: pseudo-entanglement ABV}) implies the existence of EFI pairs (Definition~\ref{def: EFI pair}).
We do this by giving an explicit construction of an efficient algorithm that prepares an EFI pair, based on the two families in the pseudoentanglement structure, \(\{k,\psi_{AB}^k\}_k\) and \(\{k,\phi_{AB}^k\}_k\), assuming that the states in both families are preparable in polynomial time.

This result, stated in Theorem~\ref{thm:pe_to_efi}, draws its technical complexity mostly from two preliminary results.
First, Lemma~\ref{lemma: connection of td and cd}, which connects the entanglement gap in the pseudoentanglement structure with the trace distance between the states in the families.
Second, Lemma~\ref{lemma: td of mixtures} which relates the trace distance between the different quantum states in two families with the trace distance between the averages over the families. 
Before demonstrating Lemma~\ref{lemma: td of mixtures}, we state the auxiliary Lemma~\ref{lemma: lemma 4}, which will be used in its proof.
\begin{lemma}[{\cite[Lemma 4]{BK02}}]\label{lemma: lemma 4}
    Let $M$ be a positive semidefinite matrix of the form
\begin{equation*}
    M = \begin{pmatrix} a & b^\dagger \\ b & c \end{pmatrix},
\end{equation*}
where $a, b, c$ are matrices.
Write the block decomposition of $\sqrt M$ as
\begin{equation*}
    \sqrt{M} = \begin{pmatrix} x & y^\dagger \\ y & z\end{pmatrix}.
\end{equation*}
Then, $\|y\|_2^2 \leqslant\|b\|_1$.
\end{lemma}

\begin{lemma}[From pairwise distances to distance of averages]\label{lemma: td of mixtures}
    Let \(\lambda\in\mathbb{N}\), $\kappa,m:\mathbb N\rightarrow\mathbb N$ be polynomially bounded functions, and   
    $\{k,\psi_{AB}^k\}_{k\in\{0,1\}^{\kappa(\lambda)}}$ and $\{k,\phi_{AB}^k\}_{k\in\{0,1\}^{\kappa(\lambda)}}$ be two families of \(m(\lambda)\)-qubit bipartite quantum states and corresponding keys.
    Assume that there exists some positive $p \in\operatorname{Poly}$ such that $\frac{1}{2}\|\psi_{AB}^k-\phi_{AB}^{k^\prime}\|_1\geq{1}/{p(\lambda)}$ for all pairs of $k,k^\prime\in\{0,1\}^{\kappa(\lambda)}$. Then, there exists a positive $q\in\operatorname{Poly}$ such that there exists $\varepsilon\in\operatorname{Negl}$ for which
    \begin{equation*}
        \frac{1}{2}\left\|\mathbb E_k\Big[{\psi^k}_{AB}^{\otimes q(\lambda)}\Big] -\mathbb E_{k^\prime}\Big[{\phi^{k^\prime}}_{AB}^{\otimes q(\lambda)}\Big] \right\|_1 \geq1-\varepsilon(\lambda).
    \end{equation*}
\end{lemma}
\begin{proof}
    We divide the proof into three steps. In the first step, we reduce the claim of the lemma to the task of distinguishing states from two different families. In the second and third steps, using the Gram matrix approach inspired by~\cite{BK02,M07,M19}, we show that the probability of guessing the wrong family when using a PGM~\cite{HJ96,HW94} can be made negligible with an appropriate choice of a polynomial  $q:\mathbb N\rightarrow\mathbb N$.
    As in Definition~\ref{def: pseudo-entanglement ABV}, we denote the averages over the families for each \(\lambda\) as $\overline{\psi^{\otimes q(\lambda)}_{{A}{B}}}^\lambda := \mathbb E_k\left[{\psi_{AB}^k}{}^{\otimes q(\lambda)}\right]$ and $\overline{\phi^{\otimes q(\lambda)}_{{A}{B}}}^\lambda:= \mathbb E_k\left[{\phi_{AB}^k}{}^{\otimes q(\lambda)}\right]$.
    Additionally, without loss of generality, we associate the bit-string keys with natural numbers, i.e.,
        \(K^\lambda = \{0,1\}^{\kappa(\lambda)} \equiv \{1,2,\dots, 2^{\kappa(\lambda)}\} = [2^{\kappa(\lambda)}]\).

    \medskip
    {\textit{TD to distinguishing:}} We start by mapping the consequent of the implication to an equivalent statement, from the distance between \(\overline{\psi^{\otimes q(\lambda)}_{{A}{B}}}^\lambda\) and \(\overline{\phi^{\otimes q(\lambda)}_{{A}{B}}}^\lambda\) to distinguishing the states.
    Consider the Holevo-Helstrom bound (Proposition~\ref{prop: holevo-helstrom}) 
    \begin{equation}
        p_{\mathrm{succ}}\left(\overline{\psi^{\otimes q(\lambda)}_{{A}{B}}}^\lambda , \overline{\phi^{\otimes q(\lambda)}_{{A}{B}}}^\lambda \right) = \frac{1}{2}\left(1+\frac{1}{2}\left\|\overline{\psi^{\otimes q(\lambda)}_{{A}{B}}}^\lambda -\overline{\phi^{\otimes q(\lambda)}_{{A}{B}}}^\lambda \right\|_1\right),
    \end{equation}
    where $p_{\mathrm{succ}}(\rho,\sigma)$ denotes the success probability in distinguishing $\rho$ and $\sigma$
    using a binary POVM $\{\Lambda, \mathbb I - \Lambda\}$ on a source emitting the two states with equal prior probabilities (\(p=1/2\)). Rearranging the terms yields
    \begin{equation}\label{eq: direct proof. trace distance HH}
        \frac{1}{2}\left\|\overline{\psi^{\otimes q(\lambda)}_{{A}{B}}}^\lambda  - \overline{\phi^{\otimes q(\lambda)}_{{A}{B}}}^\lambda \right\|_1 = 2\left(p_{\mathrm{succ}}\left(\overline{\psi^{\otimes q(\lambda)}_{{A}{B}}}^\lambda , \overline{\phi^{\otimes q(\lambda)}_{{A}{B}}}^\lambda  \right)-\frac{1}{2}\right).
    \end{equation}
    Then, let $\tilde p_{\mathrm{succ}}$ be some suboptimal success probability of distinguishing the states with an associated error probability 
    \begin{equation}
        \tilde p_{\mathrm{err}}\left(\overline{\psi^{\otimes q(\lambda)}_{{A}{B}}}^\lambda ,\overline{\phi^{\otimes q(\lambda)}_{{A}{B}}}^\lambda \right):=1-\tilde p_{\mathrm{succ}}\left(\overline{\psi^{\otimes q(\lambda)}_{{A}{B}}}^\lambda ,\overline{\phi^{\otimes q(\lambda)}_{{A}{B}}}^\lambda \right).
    \end{equation}
    Since we have that $p_{\mathrm{succ}}\left(\overline{\psi^{\otimes q(\lambda)}_{{A}{B}}}^\lambda ,\overline{\phi^{\otimes q(\lambda)}_{{A}{B}}}^\lambda \right) \geq \tilde p_{\mathrm{succ}}\left(\overline{\psi^{\otimes q(\lambda)}_{{A}{B}}}^\lambda ,\overline{\phi^{\otimes q(\lambda)}_{{A}{B}}}^\lambda \right)$, Equation~\eqref{eq: direct proof. trace distance HH} becomes
    \begin{equation}\label{eq: direct proof. lower bound on td HH}
         \frac{1}{2}\left\|\overline{\psi^{\otimes q(\lambda)}_{{A}{B}}}^\lambda  - \overline{\phi^{\otimes q(\lambda)}_{{A}{B}}}^\lambda \right\|_1\geq 2\left(\frac{1}{2}-\tilde p_\mathrm{err}\left(\overline{\psi^{\otimes q(\lambda)}_{{A}{B}}}^\lambda ,\overline{\phi^{\otimes q(\lambda)}_{{A}{B}}}^\lambda \right)\right).
    \end{equation}
    Hence, to prove the lemma, given $\overline{\psi^{\otimes q(\lambda)}_{{A}{B}}}^\lambda $ and $\overline{\phi^{\otimes q(\lambda)}_{{A}{B}}}^\lambda $, it is sufficient to provide an upper bound on the error of distinguishing the states.
    In turn, this translates to a new experiment in which a state is sampled at random from either the family $\{{\psi^{k\, \otimes q(\lambda)}}\}_{k\in K^\lambda}$ or $\{{\phi^{k\,\otimes q(\lambda)}}\}_{k\in K^\lambda}$, and asking which family the chosen state was sampled from.

    \medskip
    {\textit{Bound \(p_\mathrm{err}\) on single copies:}} Consider the problem of discriminating families of states consisting of only a single copy of the initial states from $\{{\psi_{AB}^k}\}_{k\in K^\lambda}$ and $\{{\phi_{AB}^k}\}_{k\in K^\lambda}$.
    Let $\Sigma^\lambda$ denote the joint family of all states in both families
     \begin{equation}
             \Sigma^\lambda := \{\sigma_{AB}^i\}_{i\in [2^{\kappa(\lambda)+1}]} := \{\psi_{AB}^k\}_{k\in K^\lambda} \cup \{\phi_{AB}^k\}_{k\in K^\lambda}.
     \end{equation}
     We now employ the PGM as the distinguisher between all states in $\Sigma^\lambda$. 
     For completeness, we  recall the construction of the PGM, which we denote as $M^\lambda$. 
     For each $\lambda\in\mathbb N$, $M^\lambda:=\{\mu^i\}_{i=1}^{2^{\kappa(\lambda)+1}}$, where 
     \begin{equation}
         \mu^i := S^{-1/2} \sigma^i_{AB} S^{-1/2},
     \end{equation}
     for \(i\in[2^{\kappa(\lambda)+1}]\), with $S := \sum_{i=1}^{2^{\kappa(\lambda)+1}}\sigma_{AB}^i$. 
     Here, $S^{-1/2}$ denotes the pseudo-inverse of the square root of $S$.
     Notice that since all $\sigma_{AB}^i$ are quantum states, they are positive semidefinite operators, hence, $S$ is positive semidefinite as well. We note that in general the sum of all effect operators $\mu^i$, $i\in[2^{\kappa(\lambda)+1}]$, does not equal the identity operator on the corresponding Hilbert space, since $S$ may not have the full support. However, without loss of generality, in all of the following we may assume that $\sum_i\mu^i=\mathbb I$ for all $\lambda$ as one can always complete the PGM to a measurement. For example, one can modify the PGM as follows: $\tilde\mu^1:=\mu^1+(\mathbb I-\Pi_{\operatorname{supp}S})$ and $\tilde\mu^i:=\mu^i$ for all $ i\neq1$, where $\Pi_{\operatorname{supp}S}$ is the projector on the support of $S$. Completing the PGM in this way does not affect any of the following success or error probabilities, since all states $\sigma^i_{AB}$ act on the support of $S$. 
     For a detailed discussion on the properties of the PGM, we refer to~\cite{BK02}.

   We now define the probability of error of the PGM for distinguishing between two families of states $\{{\psi_{A B}^k}\}_{k\in K^\lambda}$ and $\{{\phi_{AB}^k}\}_{k\in K^\lambda}$.
     Let $J^\lambda_b$ (\(b\in\{0,1\}\)) be the set of indices in \([2^{\kappa(\lambda)+1}]\) corresponding to the states in $\{{\psi_{AB}^k}\}_{k\in K^\lambda}$ and $\{{\phi_{AB}^k}\}_{k\in K^\lambda}$, i.e.,
     \begin{align}
        J^\lambda_0 &:=\left\{j\in[2^{\kappa(\lambda)+1}]:\sigma^j\in \{{\psi_{AB}^k}\}_{k\in K^\lambda}\right\} \subset [2^{\kappa(\lambda) + 1}],\\
        J^\lambda_1 &:=\left\{j\in[2^{\kappa(\lambda)+1}]:\sigma^j\in \{{\phi_{AB}^k}\}_{k\in K^\lambda}\right\} \subset [2^{\kappa(\lambda) + 1}].        
    \end{align}
    Then, let the probability of error for detecting states in the wrong family using the PGM $M^\lambda$ be defined as
    \begin{equation}\label{eq: direct proof. p_err}
         p^{(b)}_\mathrm{err(PGM)}(M^\lambda):=  \frac{1}{|J^\lambda_b|}\sum_{j\in J^\lambda_b}\sum_{i\notin J^\lambda_b}\operatorname{Tr}(\mu^i\sigma^j_{AB}),
    \end{equation}
   for \(\mu^i\in M^\lambda\) and $b\in\{0,1\}$.

   To upper-bound $p^{(b)}_\mathrm{err(PGM)}(M^\lambda)$, we use the Gram matrix approach similar to~\cite{BK02,M07,M19}. 
   First, consider the spectral decomposition of each state $\sigma_{AB}^i = \sum_j \lambda^{i}_j\ketbra{\chi^{i}_j}{\chi^{i}_j}$. 
   Second, for all pairs $i,j\in[2^{\kappa(\lambda)+1}]$ define a matrix $G^{(ij)}$  as
   \begin{equation}
       G^{(ij)}:=\sum_{k=1}^{\operatorname{rank}(\sigma^i_{AB})}\sum_{\ell=1}^{\operatorname{rank}(\sigma^j_{AB})}\sqrt{\lambda^i_{k}}\sqrt{\lambda^j_{l}} \braket{ \chi^i_{k}}{ \chi^j_{\ell}} \ketbra{k}{\ell},
   \end{equation}
   where $|k\rangle,|\ell\rangle$ are the elements of the canonical basis. 
   Notice that, for all indices $i,j\in [2^{\kappa(\lambda)+1}]$,  ${G^{(ji)}}^\dagger = G^{(ij)}$ since
   \begin{equation}\label{eq: direct proof. hermitatian of a block}
       G^{(ij)} = \sum_{k,\ell}\sqrt{\lambda^i_{k}}\sqrt{\lambda^j_{\ell}} \braket{ \chi^i_{k}}{\chi^j_\ell} \ketbra{k}{\ell} = \sum_{k,\ell} \sqrt{\lambda^i_{k}}\sqrt{\lambda^j_{\ell}} \;\overline{\braket{\chi^j_{\ell}}{\chi^i_{k}}} \left(\ketbra{\ell}{k}\right)^\dagger = {G^{(ji)}}^\dagger. 
   \end{equation}
   The Gram matrix $G^\lambda$ is then defined as $G^\lambda := (G^{(ij)})_{i,j=1}^{2^{\kappa(\lambda)+1}}$. 
   We remark that the resulting matrix is Hermitian and positive semidefinite with square root in block form~\cite{M07,M19}. 
   The structure of $G^\lambda$ has the form
   \begin{align}\begin{split}
       G^\lambda &:= {\footnotesize \begin{pNiceArray}{cccc|cccc}
    G^{(1,1)} & \cdots & G^{(1,2^{\kappa(\lambda)})} && G^{(1,2^{\kappa(\lambda)}+1)} & \cdots & G^{(1,2^{\kappa(\lambda)+1})}\\
   G^{(2,1)}  & \cdots & G^{(2,2^{\kappa(\lambda)})} && G^{(2,2^{\kappa(\lambda)}+1)} & \cdots & G^{(2,2^{\kappa(\lambda)+1})} \\
   \vdots & \vdots  & \vdots && \vdots & \vdots & \vdots  \\ 
   G^{(2^{\kappa(\lambda),1})} & \cdots & G^{(2^{\kappa(\lambda)},2^{\kappa(\lambda)})} && G^{(2^{\kappa(\lambda)},2^{\kappa(\lambda)}+1)} & \cdots & G^{(2^{\kappa(\lambda)},2^{\kappa(\lambda)+1})}
   \\
  \hline
  G^{(2^{\kappa(\lambda)}+1,1)} & \cdots & G^{(2^{\kappa(\lambda)}+1,2^{\kappa(\lambda)})} && G^{(2^{\kappa(\lambda)}+1,2^{\kappa(\lambda)}+1}) & \cdots & G^{(2^{\kappa(\lambda)}+1,2^{\kappa(\lambda)+1})}
  \\
  \vdots  & \vdots & \vdots &&  \vdots & \vdots & \vdots  \\ 
G^{(2^{\kappa(\lambda)+1},1)} & \cdots & G^{(2^{\kappa(\lambda)+1},2^{\kappa(\lambda)})} && G^{(2^{\kappa(\lambda)+1},2^{\kappa(\lambda)}+1)}  & \cdots & G^{(2^{\kappa(\lambda)+1},2^{\kappa(\lambda)+1})}
    \end{pNiceArray}}
    \\
    &=:  {\footnotesize\begin{pNiceArray}{c|c}
        G_{00} & G_{01} \\
        \hline 
        G_{10} & G_{11} 
    \end{pNiceArray}
    = \begin{pNiceArray}{c|c}
        G_{00} & G_{10}^\dagger \\
        \hline 
        G_{10} & G_{11} 
    \end{pNiceArray}.}
    \label{eq: direct proof. daggered block decomp}
   \end{split}
   \end{align}
   
   Recall the definition of $p_{\mathrm{err(PGM)}}^{(b)}$ in Equation~\eqref{eq: direct proof. p_err}.
   It can be rewritten in terms of the Schatten 2-norm (Frobenius norm) of the corresponding blocks of $\sqrt{G}^\lambda$ using \(
       \operatorname{Tr}(\mu^i\sigma^j_{AB}) = \|\sqrt{G}^{(ij)}\|_2^2
    \),
    for all
    \(i,j\in [2^{\kappa(\lambda)+1}]\)~\cite[Section 2]{M19}:
    \begin{align}
       p_{\mathrm{err(PGM)}}^{(b)}(M^\lambda) &= \frac{1}{|J^\lambda_b|}\sum_{j\in J^\lambda_b}\sum_{i\notin J^\lambda_b}\operatorname{Tr}(\mu^i\sigma^j_{AB}) =\frac{1}{|J^\lambda_b|}\sum_{j\in J^\lambda_b}\sum_{i\notin J^\lambda_b}\|\sqrt{G}^{(ij)}\|_2^2. \label{eq: direct proof. p_err second form}
   \end{align}
    It is thus sufficient to upper-bound the sum of \(\|\sqrt{G}^{(ij)}\|_2^2\). 
    Since $G^\lambda$ can be written as a \(2\)-by-\(2\) block matrix and it is positive semidefinite, from Lemma~\ref{lemma: lemma 4},
    \begin{equation}\label{eq: direct proof. bound on blocks of G}
        \|\sqrt G_{10}\|_2^2 \leq \|G_{10}\|_1,
    \end{equation}
    where $\sqrt{G}_{10}$ is the corresponding block of $\sqrt{G}^\lambda$.
    We note that such a form exists due to $G^\lambda$ being Hermitian and positive semidefinite, which, in turn, implies that $\sqrt{G}^\lambda$ is positive semidefinite and Hermitian~\cite{KW24}. 
    For the explicit form of $\sqrt{G}^\lambda$, we refer the reader to~\cite{M19}.
    Note that similar results hold for $G_{01}$ since $G_{01} = G_{10}^\dagger$ and Hermitian conjugation does not affect Schatten \(p\)-norms~\cite{KW24}. 

   Now, consider the left-hand side of Equation~\eqref{eq: direct proof. bound on blocks of G}. 
   Since the squared Frobenius norm (Definition~\ref{def:shattenpnorm}) of an operator is the sum of the squares of its entries,
   \begin{equation} \label{eq: direct proof. sqrt G_10 bound}
       \|\sqrt G_{10}\|_2^2 = \sum_{i\in J^\lambda_1}\sum_{j\in J^\lambda_0}\|\sqrt{G}^{(ij)}\|_2^2.
   \end{equation}
   For the right-hand side of Equation~\eqref{eq: direct proof. bound on blocks of G},
   \begin{align}\begin{split}
       \left\| G_{10}  \right\|_1 = \left\|\sum_{i\in J^\lambda_1}\sum_{j\in J^\lambda_0}\ketbra{i}{j}\otimes G^{(ij)} \right\|_1 
       &\leq \sum_{i\in J^\lambda_1}\sum_{j\in J^\lambda_0}\left\|\ketbra{i}{j}\otimes G^{(ij)}\right\|_1
       \\
       & = \sum_{i\in J^\lambda_1}\sum_{j\in J^\lambda_0} \|G^{(ij)}\|_1\|\ketbra{i}{j}\|_1
       \\
       &=\sum_{i\in J^\lambda_1}\sum_{j\in J^\lambda_0} \|G^{(ij)}\|_1, \label{eq: direct proof. G_10 bound}
   \end{split}\end{align}
   where the inequality follows from the triangle inequality property of the trace norm, and the following equality results from its multiplicativity under the tensor product. 

   Recomputing Equation~\eqref{eq: direct proof. bound on blocks of G} from Equations~\eqref{eq: direct proof. G_10 bound} and~\eqref{eq: direct proof. sqrt G_10 bound} yields
   \begin{equation}\label{eq: direct proof. bound on p_err (1)}
       p_{\text{err(PGM)}}^{(1)}(M^\lambda) = \frac{1}{|J^\lambda_1|}\sum_{i\in J^\lambda_1}\sum_{j\in J^\lambda_0}\|\sqrt{G}^{(ij)}\|_2^2 \leq \frac{1}{|J^\lambda_1|}\sum_{i\in J^\lambda_1}\sum_{j\in J^\lambda_0} \|G^{(ij)}\|_1.
   \end{equation}
   Finally, for all \(i,j\in[2^{\kappa(\lambda)+1}]\), as in~\cite{M19},
   \begin{align}\begin{split}\label{eq: direct proof. 1-norm G}
       \|G^{(ij)}\|_1 =\left\|\left(\sum_k \sqrt{\lambda^i_{ k}}\left|\chi^i_{ k}\right\rangle\left\langle\chi^i_{ k}\right|\right)\left(\sum_\ell \sqrt{\lambda^j_{\ell}}\left|\chi^j_{\ell}\right\rangle\left\langle\chi^j_{\ell}\right|\right)\right\|_1&=\left\|\sqrt{\sigma_{AB}^i} \sqrt{\sigma_{AB}^j}\right\|_1 \\&=\sqrt{\operatorname{F}}(\sigma_{AB}^i,\sigma_{AB}^j),
   \end{split}\end{align}
   where $\sqrt{\operatorname{F}}(\sigma_{AB}^i,\sigma_{AB}^j)$ denotes the square root of the fidelity (Definition~\ref{def: fidelity}).
    Thus,
    \begin{equation}\label{eq: direct proof. bound on p_err. final}
       p_{\mathrm{err(PGM)}}^{(1)}(M^\lambda) \leq \frac{1}{|J^\lambda_1|}\sum_{i\in J^\lambda_1}\sum_{j\in J^\lambda_0} \sqrt{\operatorname{F}}\left(\sigma_{AB}^i,\sigma_{AB}^j\right) = \frac{1}{|J^\lambda_1|}\sum_{i=1}^{2^{\kappa(\lambda)}} \sum_{j=1}^{2^{\kappa(\lambda)}}\sqrt{\operatorname{F}}\left(\psi^i_{AB}, \phi^j_{AB}\right).
   \end{equation}
   The same may be obtained for $p_{\mathrm{err(PGM)}}^{(0)}(M^\lambda)$, replacing $G_{10}$ and $\sqrt G_{10}$ with $G_{01}$ and $\sqrt G_{01}$, respectively,
   \begin{equation}
       p_{\mathrm{err(PGM)}}^{(0)}(M^\lambda) \leq\frac{1}{|J^\lambda_0|} \sum_{i=1}^{2^{\kappa(\lambda)}} \sum_{j=1}^{2^{\kappa(\lambda)}}\sqrt{\operatorname{F}}\left(\psi^j_{AB}, \phi^i_{AB}\right).
   \end{equation}

    Hence, for a single copy, the probability of the PGM outputting the wrong answer may be bounded as
   \begin{align}\begin{split}
       p_{\mathrm{err(PGM)}}(M^\lambda) &= \frac{1}{2}\left[ p_{\mathrm{err(PGM)}}^{(0)}(M^\lambda) + p_{\mathrm{err(PGM)}}^{(1)}(M^\lambda)\right] \\
       &\leq \frac{1}{|J^\lambda_0|}\sum_{i=1}^{2^{\kappa(\lambda)}} \sum_{j=1}^{2^{\kappa(\lambda)}}\sqrt{\operatorname{F}}\left(\psi^i_{AB}, \phi^j_{AB}\right)
       \\
       &\leq2^{\kappa(\lambda)}\sqrt{\left(1-\left(\frac{1}{p(\lambda)}\right)^2\right)}\label{eq: direct proof. bound on total p_{err}},
\end{split}\end{align}
where the first inequality follows by plugging $|J_0^\lambda|=|J_1^\lambda|=2^{\kappa(\lambda)}$. The last inequality follows by the assumption of the lemma and the Fuchs-van de Graaf inequality (Proposition~\ref{prop: fvdg}), since, if for all \(k,k^\prime\in K^\lambda\) \(
    \frac{1}{2}\|\psi^k_{AB} - \phi_{AB}^{k^\prime}\|_1\geq {1}/{p(\lambda)}\), then \( \operatorname{F}(\psi^k_{AB},\phi_{AB}^{k^\prime}) \leq 1 - \left({1}/{p(\lambda)}\right)^2
\).

   \medskip
   {\textit{Bound \(p_\mathrm{err}\) and TD on polynomial copies:}} 
    Notice that, on its own, the bound in Equation~\eqref{eq: direct proof. bound on total p_{err}} is not negligible in $\lambda$.
    However, by extending the above single-copy discrimination result to  $q \in\operatorname{Poly}$ copies of the states,  $\{{\psi_{AB}^{k\, \otimes q(\lambda)}}\}_k$ and $\{{\phi_{AB}^{k\, \otimes q(\lambda)}}\}_k$, an analogous derivation yields the PGM $ N^{\lambda}$ such that
\begin{align}\begin{split}
    p_{\mathrm{err(PGM)}}( N^{\lambda}) &\leq \frac{1}{|J^\lambda_0|}\sum_{i=1}^{2^{\kappa(\lambda)}} \sum_{j=1}^{2^{\kappa(\lambda)}}\sqrt{\operatorname{F}}\left({\psi^i}^{\otimes q(\lambda)}_{AB}, {\phi^j}^{\otimes q(\lambda)}_{AB}\right)\\
    &\leq 2^{\kappa(\lambda)}\left(1-\left(\frac{1}{p(\lambda)}\right)^2\right)^{\frac{q(\lambda)}{2}}  
    \\
    & \leq 2^{\kappa(\lambda)}e^{-\frac{q(\lambda)}{2p^2(\lambda)}} \\
    &= e^{\kappa(\lambda)\operatorname{ln}2}e^{-\frac{q(\lambda)}{2p^2(\lambda)}} \\
    &= e^{\left(-\frac{q(\lambda)}{2p^2(\lambda)} + \kappa(\lambda)\operatorname{ln}2\right)},
\end{split}\end{align}
where the last inequality follows from \((1+x)^r\leq e^{rx}\) (\(x\geq-1, r\geq0\)).
Therefore, taking the polynomial $q(\lambda) = 2 p^2(\lambda)(\kappa(\lambda)+\lambda)$ is sufficient to make the resulting $p_{\mathrm{err(PGM)}}( N^\lambda)$ negligible in $\lambda$. 

\smallskip

Finally, construct a binary POVM $\Lambda^\lambda_0=\sum_{i\in J^\lambda_0}\nu^{i}$ and $\Lambda^\lambda_1 = \sum_{i\in J^\lambda_1}\nu^{i}$, where $\nu^{i}\in N^{\lambda}$ for all \(i\in[2^{\kappa(\lambda)+1}]\). 
That is, take the sum of the effect operators of the PGM $N^{\lambda}$ corresponding to the states in the families $\{{\psi_{AB}^{k\, \otimes q(\lambda)}}\}_k$ and $\{{\phi_{AB}^{k\,\otimes q(\lambda)}}\}_k$, mapping each to $\Lambda_0$ and $\Lambda_1$, respectively. 
Note that this indeed constitutes a valid measurement, since $\Lambda^\lambda_0 + \Lambda^\lambda_1 = \sum_{i=1}^{2^{\kappa(\lambda)+1}}\nu^{i} = \mathbb I^\lambda$, because the constructed PGM is a measurement.
Positivity of $\Lambda^\lambda_0$ and $\Lambda^\lambda_1$ is a consequence of the positivity of $ \nu^{i}$ for all $i\in[2^{\kappa(\lambda)+1}]$. 
Hence, plugging $p_\mathrm{err(PGM)}( N^{\lambda})$ as the suboptimal error probability considered in Equation~\eqref{eq: direct proof. lower bound on td HH} yields the desired result
\begin{equation}
    \frac{1}{2}\left\|\overline{\psi^{\otimes q(\lambda)}_{{A}{B}}}^\lambda  - \overline{\phi^{\otimes q(\lambda)}_{{A}{B}}}^\lambda \right\|_1\geq 2\left(\frac{1}{2} - \varepsilon(\lambda)\right) \geq 1 -\tilde \varepsilon(\lambda),
\end{equation}
 where $\varepsilon,\tilde \varepsilon\in\operatorname{Negl}$. This establishes the proof of the lemma.
\end{proof}

We recall that, in Theorem~\ref{thm:pe_to_efi}, and in the remainder of the paper, we assume that both families in the definitions of pseudoentanglement are efficiently preparable by a uniform QPT algorithm. 
While the definitions of pseudoentanglement for mixed states in the computational setting (Definition~\ref{def: pseudo-entanglement ABV} and~\ref{def: pseudo-entanglement GE}) do not require the reference family to be efficiently preparable, in practice all proposed constructions of pseudoentanglement have this property~\cite{ABV23,GE24,ABF24,BF24}.

\begin{theorem}[Fully-computational pseudoentanglement $\Rightarrow$ EFI pairs]\label{thm:pe_to_efi}
    Let \(\lambda\in\mathbb N\), $\varepsilon:\mathbb N\rightarrow[0,1]$ be a vanishing at infinity function.
    Let $\{k,\psi_{AB}^k\}_{k}$, $\{k,\phi_{AB}^k\}_{k}$ be the two families in the fully-computational pseudoentanglement structure as in Definition~\ref{def: pseudo-entanglement ABV}, with $c(\lambda)<d(\lambda)$. Assume there also exist a uniform QPT algorithm $\mathcal A$ taking as input \(k\) and outputting $\psi^k_{AB}$ and $\mathcal B$ taking as input \(k\) and outputting $\phi^k_{AB}$.
    Then, there exists an EFI pair as in Definition~\ref{def: EFI pair}.
\end{theorem}
\begin{proof}
To prove this theorem, we individually show the three properties of an EFI pair.
Efficient preparation and computational indistinguishability are both trivial consequences of the assumptions of the theorem. 
To show the statistical distance requirement, we leverage the results of Lemma~\ref{lemma: td of mixtures} and Lemma~\ref{lemma: connection of td and cd}.

\medskip
\textit{Efficient preparation:} By assumption, $\{k,\psi_{AB}^k\}_k$ and $\{k,\phi_{AB}^k\}_k$ are efficiently preparable by two uniform QPT algorithms, $\psi_{AB}^k \leftarrow \mathcal{A}(1^\lambda,k)$ and  $\phi_{AB}^k \leftarrow \mathcal{B}(1^\lambda,k)$.
Then, sample a random \(k\in K^\lambda\), use \(\mathcal{A}\) and \(\mathcal{B}\) to prepare \(q\in \operatorname{Poly}\) copies of the corresponding state, and finally forget \(k\).
The EFI pair \((\rho_0^\lambda,\rho_1^\lambda)\) is given as the tensor product of the \(q(\lambda)\) copies of the outputs of \(\mathcal{A}\) and \(\mathcal{B}\),
\begin{align}
    \label{eq:THM1-eficonstr0}
    \rho_0^\lambda &= \overline{\psi^{\otimes q(\lambda)}_{{A}{B}}}^\lambda \leftarrow \bigotimes_{1}^{q(\lambda)} \mathcal{A}(1^\lambda, k);\\
    \label{eq:THM1-eficonstr1}
    \rho_1^\lambda &= \overline{\phi^{\otimes q(\lambda)}_{{A}{B}}}^\lambda \leftarrow \bigotimes_{1}^{q(\lambda)} \mathcal{B}(1^\lambda, k).
\end{align}
Finally, let the uniform QPT algorithm \(\rho_b^\lambda \leftarrow \mathcal{C}(1^\lambda, b)\) be the EFI pair generator that outputs \(\rho_0^\lambda\) when \(b=0\) and \(\rho_1^\lambda\) when \(b=1\).
This ensures property \textit{(i)} in Definition~\ref{def: EFI pair} for the pair $(\rho_0^\lambda, \rho_1^\lambda)$. 

\medskip
\textit{Computational indistinguishability:}
The two families $\{k,\psi_{AB}^k\}_k$ and $\{k,\phi_{AB}^k\}_k$ are computationally indistinguishable, as required by property  \textit{(iii)} of Definition~\ref{def: pseudo-entanglement ABV}, i.e., $\overline{\psi^{\otimes q(\lambda)}_{{A}{B}}}^\lambda \approx \overline{\phi^{\otimes q(\lambda)}_{{A}{B}}}^\lambda$.
This directly translates to the two elements of the EFI pair in Equations~\eqref{eq:THM1-eficonstr0} and~\eqref{eq:THM1-eficonstr1}, as required by Definition~\ref{def: EFI pair} \textit{(iii)}.

\medskip
\textit{Statistical distance:}
To demonstrate that the trace distance between \(\rho_0^\lambda\) and \(\rho_1^\lambda\) is large (\(\operatorname{TD}(\rho_0^\lambda,\rho_1^\lambda)\geq 1/p(\lambda)\)), we first prove that an entanglement gap (\(c<d\)) implies that, for large enough $\lambda$, for every \(k,k'\in K^\lambda\) the trace distance between the corresponding states \(\psi_{AB}^k,\phi_{AB}^{k'}\) is large.
Then, we use Lemma~\ref{lemma: td of mixtures} to go from this pairwise distance to the distance between the two uniform mixtures over all \(k,k'\).

For this, we instead start with the equivalent contrapositive of the first argument.
If the trace distance is smaller than the inverse of any polynomial even for a single pair of elements, then there cannot be an entanglement gap between the families $\{k,\psi_{AB}^k\}_k$ and $\{k,\phi_{AB}^k\}_k$.
Formally, given $\hat E^\varepsilon_C (\{k,\psi_{AB}^k\}_k) \leq c$ and $\hat E^\varepsilon_D (\{k,\phi_{AB}^k\}_k) \geq d$, if for all $p\in\operatorname{Poly}$, for all \(\lambda \in \mathbb{N}\), there exists a \(\lambda^\star\geq \lambda\) and a pair of keys \(k,k' \in K^{\lambda^\star}\), such that 
\begin{equation}\label{eq: negl st dist}
         \frac{1}{2}\|\psi_{AB}^k - \phi_{AB}^{k^\prime}\|_1<\frac{1}{p(\lambda^\star)},
\end{equation}
then there exist infinitely many \(\mu\in\mathbb{N}\) such that \(c(\mu) \geq d(\mu)\).

First, take any non-constant positive polynomial $p\in\operatorname{Poly}$. 
Then, for each \(\lambda\in\mathbb{N}\), consider the smallest \(\lambda^\star \geq \lambda\) and the corresponding \(k,k'\in K^{\lambda^\star}\) and \(\psi_{AB}^k,\phi_{AB}^{k'}\), that fulfills Equation~\eqref{eq: negl st dist} (if there is more than one pair \((k,k')\), pick the one with the smallest keys).
Construct the families $\{{\psi^\star}^\lambda_{AB}\}_{\lambda}$ and $\{{\phi^\star}_{AB}^{\lambda}\}_{\lambda}$ as follows. (This construction is required since \(c,d\) are only defined for families and not single states.)
For \(\lambda=1\), let \({\psi^\star}^1_{AB} = \psi_{AB}^{0}\) and \({\phi^\star}^1_{AB} = \phi_{AB}^0\) (wlog, \(k,k'=0\)); then for \(\lambda>1\)
\begin{align}\label{eq: step-wise}
{\psi^\star}^\lambda_{AB} =
\begin{cases}
     \psi_{AB}^k & \text{if}\quad \lambda = \lambda^\star,\\ 
    {\psi^\star}^{\lambda-1}_{AB} & \text{if}\quad \lambda < \lambda^\star.
\end{cases}\qquad 
{\phi^\star}^\lambda_{AB} =
\begin{cases}
     \phi_{AB}^{k'} & \text{if}\quad \lambda = \lambda^\star,\\ 
    {\phi^\star}^{\lambda-1}_{AB} & \text{if}\quad \lambda < \lambda^\star.
\end{cases}
\end{align}
For convenience, let $\Lambda^\star\subset\mathbb N$ denote the set of ``bad indices'', consisting of all $\lambda>1$ for which Equation~\eqref{eq: negl st dist} holds. Namely, 
\begin{equation}\label{eq: bad indices}
    \Lambda^\star:=\left\{\lambda>1\,|\,\exists k,k^\prime\in K^\lambda,\, \frac{1}{2}\|\psi_{AB}^k - \phi_{AB}^{k^\prime}\|_1<\frac{1}{p(\lambda)} \right\}.
\end{equation}
Therefore, in the definitions of one-shot computational distillable entanglement and cost (Definitions~\ref{def: comp dist} and~\ref{def: comp cost}), the same construction as in Equation~\eqref{eq: step-wise} applies to the families of distillation and dilution channels, the functions \(n^\star_A,n^\star_B\) (which are still polynomially bounded), the upper (resp.\ lower) bound \(c^\star\) (resp.\ \(d^\star\)) and the error of distillation $\varepsilon^\star$. The families of distillation and dilution channels are constructed by repeating the selecting procedure in Equation~\eqref{eq: step-wise}:
\begin{align}\label{eq: step-wise channels}
{\Gamma^\star}^{\lambda}_{D} =
\begin{cases}
     \Gamma^1_{D}(0,\cdot) & \text{if}\quad \lambda = 1,\\ 
     \Gamma^\lambda_{D}(k^\prime,\cdot)& \text{if}\quad \lambda \in\Lambda^\star,\\
    {\Gamma^\star}^{\lambda-1}_{D} & \text{otherwise}.
\end{cases}\qquad 
{\Gamma^\star_C}^\lambda =
\begin{cases}
     \Gamma^1_{C}(0,\cdot) & \text{if}\quad \lambda = 1,\\ 
     \Gamma^\lambda_{C}(k,\cdot)& \text{if}\quad \lambda \in\Lambda^\star,\\
    {\Gamma^\star_C}^{\lambda-1} & \text{otherwise}.
\end{cases}
\end{align}
where $\Lambda^\star\subset \mathbb N$ is defined in Equation~\eqref{eq: bad indices} and $\{\Gamma_D^\lambda\}_\lambda,\{\Gamma_C^\lambda\}_\lambda$ are efficient families of distillation and dilution channels achieving valid lower bound $d$ and upper bound $c$, respectively. 
 
The construction of the remaining functions requires a slight modification. First, we provide an explicit construction for the functions $n_A^\star,n_B^\star$. For each $\lambda\in\mathbb N$,
\begin{align}\label{eq: step-wise n_a/n_b}
n^\star_{A}(\lambda) =
\begin{cases}
    n_A(1) & \text{if}\quad \lambda = 1,
    \\ 
     n_A(\lambda) & \text{if}\quad \lambda \in \Lambda^\star,\\ 
    n^\star_A(\lambda-1) & \text{otherwise}.
\end{cases}\qquad 
n^\star_{B}(\lambda)  =
\begin{cases}
     n_B(1) & \text{if}\quad \lambda = 1,
    \\ 
     n_B(\lambda) & \text{if}\quad \lambda \in \Lambda^\star,\\ 
    n^\star_B(\lambda-1) & \text{otherwise}.
\end{cases}
\end{align}
The functions $n_A^\star,n_B^\star$ are polynomially bounded, since $n_A$ and $n_B$, which are polynomially bounded by Definition~\ref{def: pseudo-entanglement ABV}, asymptotically upper-bound $n_A^\star$ and $n_B^\star$, respectively. A similar argument is used to infer that the families of distillation and dilution  channels, \(\{{\Gamma_D^{\star\lambda}}\}_{\lambda}\) and \(\{{\Gamma_C^{\star\lambda}}\}_{\lambda}\), are efficient. Therefore, the  upper (resp.\ lower) bound \(c^\star\) (resp.\ \(d^\star\)) holds, i.e.,
\begin{equation}\label{eq:THM1-gaptau}
    \hat E^{\varepsilon^\star}_C(\{{\psi^\star}^{\lambda}_{AB}\}_\lambda)\leq c^\star \quad\text{and}\quad \hat E^{\varepsilon^\star}_D(\{{\phi^\star}^{\lambda}_{AB}\}_\lambda)\geq d^\star,
\end{equation}
where the error of distillation $\varepsilon^\star$ is defined below, and the corresponding valid upper and lower bounds, $c^\star$ and $d^\star$, are defined in a similar way as in Equation~\eqref{eq: step-wise n_a/n_b}. For each $\lambda\in\mathbb N$,
\begin{align}\label{eq: step-wise c/d}
c^\star(\lambda) =
\begin{cases}
    c(1) & \text{if}\quad \lambda = 1,
    \\ 
     c(\lambda) & \text{if}\quad \lambda \in \Lambda^\star,\\ 
    c^\star(\lambda-1) & \text{otherwise}.
\end{cases}\qquad 
d^\star(\lambda) =
\begin{cases}
    d(1) & \text{if}\quad \lambda = 1,
    \\ 
     d(\lambda) & \text{if}\quad \lambda \in \Lambda^\star,\\ 
    d^\star(\lambda-1) & \text{otherwise}.
\end{cases}
\end{align}

Second, by construction, for all elements in the families $\{{\psi^\star}^\lambda_{AB}\}_{\lambda}$ and $\{{\phi^\star}_{AB}^{\lambda}\}_{\lambda}$,
\begin{equation}\label{eq:THM1-tauTD}
    \frac{1}{2}\left\|{\psi^\star}^{\lambda}_{AB} - {\phi^\star}_{AB}^{\lambda} \right\|_1 \leq \frac{1}{p^\star(\lambda)},
\end{equation}
where $p^\star$ is defined stepwise, repeating the construction described in Equation~\eqref{eq: step-wise}. Analogously to the construction in Equation~\eqref{eq: step-wise n_a/n_b}, we define functions $p^\star$ and $\varepsilon^\star$, for each $\lambda\in\mathbb N$, as
\begin{align}\label{eq: step-wise p and eps}
 p^\star (\lambda)=
    \begin{cases}
    1 & \text{if}\quad \lambda = 1, \\
    p(\lambda) & \text{if}\quad \lambda \in \Lambda^\star,\\
      p^\star(\lambda-1) & \text{otherwise}.
    \end{cases}\qquad 
 \varepsilon^\star (\lambda)=
    \begin{cases}
    1 & \text{if}\quad \lambda = 1,\\
    \varepsilon(\lambda) & \text{if}\quad \lambda \in \Lambda^\star,\\ 
      \varepsilon^\star(\lambda-1) & \text{otherwise}.
    \end{cases}
\end{align}

Observe that the stepwise constructions of $p^\star$ and $\varepsilon^\star$ in Equation~\eqref{eq: step-wise p and eps}, while potentially slowing their respective growth and decay, preserve asymptotic behavior of the original functions $p$ and $\varepsilon$, respectively. That is, $\varepsilon^\star$ and $1/p^\star$ both converge to zero as $\lambda$ goes to infinity. Therefore, for large enough $\lambda$, by Lemma~\ref{lemma: connection of td and cd}, since the trace distance is upper-bounded by a vanishing at infinity function, the bounds fulfill \(c^\star \geq d^\star\) for \(\lambda\) larger than some \(\mu\in\mathbb{N}\).
Furthermore, from the definition of fully-computational pseudoentanglement (Definition~\ref{def: pseudo-entanglement ABV}), as a consequence of the definitions of computational distillable entanglement and cost (Definitions~\ref{def: uni comp dist},~\ref{def: uni comp cost}), it is required that the upper and lower bounds hold for all the states in the initial keyed families $\{k,\psi_{AB}^k\}_k$ and $\{k,\phi_{AB}^k\}_k$.
In particular, they must hold for the states used to define each \(\psi^\star_{AB} = \psi_{AB}^k\) and \(\phi^\star_{AB} = \phi_{AB}^{k'}\) and the respective families  \(\{{\psi^\star}^{\lambda}_{AB}\}_\lambda\) and \(\{{\phi^\star}^{\lambda}_{AB}\}_\lambda\).
Therefore, at each of the (infinitely many) values \(\lambda^\star \geq \mu\), the corresponding upper bound on \(\hat E^\varepsilon_C(\{{k,\psi}^{k}_{AB}\}_k)\leq c\) and lower bound on \(\hat E^\varepsilon_D(\{{k,\phi}^{k}_{AB}\}_k)\geq d\), must consider that 
\(c(\lambda^\star) \geq d(\lambda^\star)\), which proves the contrapositive. Therefore, the original direction holds:
if for all \(\lambda\in\mathbb{N}\), \(c(\lambda) < d(\lambda)\), then there exists a $p\in\operatorname{Poly}$ and \(\lambda_0\in\mathbb{N}\) such that for all \(\lambda\geq \lambda_0\) and for all $k,k^\prime\in\{0,1\}^{\kappa(\lambda)}$, 
\begin{equation}
    \frac{1}{2}\|\psi_{AB}^k - \phi_{AB}^{k^\prime}\|_1 \geq \frac{1}{p(\lambda)}.
\end{equation}

Finally, from Lemma~\ref{lemma: td of mixtures}, given that the families $\{k,\psi^k_{AB}\}_k$ and $\{k,\phi_{AB}^k\}_k$ have pairwise trace distance larger than some inverse polynomial, taking their averages over \(q(\lambda)\) copies (\(q\in\operatorname{Poly}\)) gives that, for \(\lambda\) large enough, there exists a \(\varepsilon\in\operatorname{Negl}\) such that
\begin{equation}
    \frac{1}{2}\left\|\overline{\psi^{\otimes q(\lambda)}_{{A}{B}}}^\lambda  - \overline{\phi^{\otimes q(\lambda)}_{{A}{B}}}^\lambda \right\|_1 \geq1-\varepsilon(\lambda).
\end{equation}

Therefore, the states in the proposed EFI pair, $\rho^\lambda_0$ and $\rho_1^\lambda$, are statistically far from each other. That is, there exists a polynomial $p\in\operatorname{Poly}$, such that, for large enough \(\lambda\), \(\operatorname{TD}(\rho_0^\lambda, \rho_1^\lambda) \geq \frac{1}{p(\lambda)}\).
This fulfills the statistical distance requirement of EFI pairs (Definition~\ref{def: EFI pair} \textit{(ii)}) and concludes the proof.
\end{proof}

\section{Equivalence Between Inefficiently-Distillable Pseudoentanglement and EFI Pairs}\label{sec: equiv}
We now demonstrate the second of our two main theorems: the equivalence between inefficiently-distillable pseudoentanglement~\cite{GE24} (Definition~\ref{def: pseudo-entanglement GE}) and EFI pairs (Definition~\ref{def: EFI pair}).
Our starting point is the implication proved in~\cite{GE24} where inefficiently-distillable pseudoentanglement is necessary for EFI pairs to exist.
This is achieved by using the existence of EFI pairs to construct a uniform QPT algorithm that outputs two families of states that are computationally indistinguishable, but that exhibit an entanglement gap.
We state their result in Lemma~\ref{lemma:efi implies pe ge} and refer to the original work for the full proof~\cite[Theorem 3.1]{GE24}.

\begin{lemma}[EFI pairs \(\Rightarrow\) Inefficiently-distillable pseudoentanglement~{\cite[Theorem 3.1]{GE24}}]\label{lemma:efi implies pe ge}
    Let $\lambda\in\mathbb N$ be a security parameter, $n:\mathbb N\to\mathbb N$ be a polynomially bounded function, and an $\varepsilon(\lambda)\in O(2^{-\lambda})$.\\
    If there exists an EFI pair outputting \(n(\lambda)\)-qubit states, then there exist two families of states, $\{\psi^\lambda_{AB}\}_\lambda$ and $\{\phi^\lambda_{AB}\}_\lambda$, consisting of $(n(\lambda)+1)$-qubit bipartite mixed states, such that $E_D^{\varepsilon}(\{\phi^\lambda_{AB}\}_\lambda)\geq 1$ and $\hat E_C^\varepsilon(\{\psi^\lambda_{AB}\}_\lambda) = 0$, and $\{{\psi_{AB}^{\lambda\,\otimes q(\lambda)}}\}_\lambda\approx \{{\phi_{AB}^{\lambda\,\otimes q(\lambda)}}\}_\lambda$ for all $q\in{\operatorname{Poly}}$.
\end{lemma}

In Lemma~\ref{lemma: pse imply efi GE}, we show that the existence of inefficiently-distillable pseudoentanglement implies the existence of EFI pairs. This effectively shows that the existence of one is necessary and sufficient for the existence of the other, which we then state in Theorem~\ref{thm:ipe_equiv_efi}.

\begin{lemma}[Inefficiently-distillable pseudoentanglement \(\Rightarrow\) EFI pairs]\label{lemma: pse imply efi GE}
    Let \(\lambda\in\mathbb N\), $\varepsilon:\mathbb N\rightarrow(0,1)$ be a vanishing at infinity function.
    Let $\{\psi_{AB}^\lambda\}_{\lambda}$, $\{\phi_{AB}^\lambda\}_{\lambda}$ be the two families in the inefficiently-distillable pseudoentanglement structure as in Definition~\ref{def: pseudo-entanglement GE}, with $c(\lambda)<d(\lambda)$. Assume there also exist a uniform QPT algorithm $\mathcal A$ generating $\{\psi^\lambda_{AB}\}_\lambda$ and a uniform QPT algorithm $\mathcal B$ generating $\{\phi_{AB}^\lambda\}_\lambda$. 
    Then, there exists an EFI pair as in Definition~\ref{def: EFI pair}.
\end{lemma}
\begin{proof}
     Similarly to the proof of Theorem~\ref{thm:pe_to_efi}, we verify each property of an EFI pair (Definition~\ref{def: EFI pair}) for the pair $(\psi^\lambda_{AB}, \phi^\lambda_{AB})$. 
     First, we address the efficiency of preparing a given pair, which is given by assumption.
     Second, we go over the computational indistinguishability of the states in the candidate pair, which follows directly.
     Third, we address the statistical distance requirement between the states $\psi^\lambda_{AB}$ and $\phi^\lambda_{AB}$, leveraging properties \textit{(i)} and \textit{(ii)} of Definition~\ref{def: pseudo-entanglement GE}.

        \medskip
        \textit{Efficient preparation:} By assumption, the families \(\{\psi_{AB}^\lambda\}_\lambda\), \(\{\phi_{AB}^\lambda\}_\lambda\) are efficiently preparable, i.e., there exist two uniform QPT algorithms \(\psi_{AB}^\lambda \leftarrow \mathcal{A}(1^\lambda)\) and \(\phi_{AB}^\lambda \leftarrow \mathcal{B}(1^\lambda)\).
        Then, let the uniform QPT algorithm \(\rho_b^\lambda \leftarrow \mathcal{C}(1^\lambda, b)\) be the EFI pair generator that calls \(\rho_0^\lambda = \psi^\lambda_{AB} \leftarrow \mathcal{A}(1^\lambda)\) when \(b=0\) and \(\rho_1^\lambda = \phi^\lambda_{AB} \leftarrow \mathcal{B}(1^\lambda)\) when \(b=1\).
        This guarantees property \textit{(i)} in Definition~\ref{def: EFI pair} for $(\rho_0^\lambda, \rho_1^\lambda)$. 

        \medskip
        \textit{Computational indistinguishability:} 
        By definition of inefficiently-distillable pseudoentanglement, the two families  $\{\psi_{AB}^\lambda\}_\lambda$ and $\{\phi_{AB}^\lambda\}_\lambda$  are computationally indistinguishable, i.e., $\{{\psi^{\lambda\,\otimes q(\lambda)}_{AB}}\}_\lambda \approx \{{\phi^{\lambda\,\otimes q(\lambda)}_{AB}}\}_\lambda$.
        This, for \(q(\lambda)=1\), directly translates to the two elements of the EFI pair as in Definition~\ref{def: EFI pair} \textit{(iii)}.

        \medskip
        \textit{Statistical distance:} 
        Since the families  \(\{\psi_{AB}^\lambda\}_\lambda\), \(\{\phi_{AB}^\lambda\}_\lambda\) can be viewed as families directly indexed by the security parameter \(\lambda\), the property \textit{(ii)} in Definition~\ref{def: EFI pair} for $(\rho_0^\lambda, \rho_1^\lambda)$ follows from the same argument as in the proof of Theorem~\ref{thm:pe_to_efi}, with a modification of replacing Lemma~\ref{lemma: connection of td and cd} with Corollary~\ref{cor: connection of td and cd_in}. 

        \smallskip
        This concludes the proof.
    \end{proof}

\begin{theorem}[Inefficiently-distillable pseudoentanglement \(\Leftrightarrow\) EFI pairs]\label{thm:ipe_equiv_efi}
    Let \(\lambda\in\mathbb N\), and a $\varepsilon:\mathbb N\rightarrow(0,1)$ with $\varepsilon(\lambda)\in O(2^{-\lambda})$.
    The existence of inefficiently-distillable $(\varepsilon,c,d)$-pseudoentanglement as in Definition~\ref{def: pseudo-entanglement GE}, with $c(\lambda)<d(\lambda)$ and such that there exist uniform QPT algorithms outputting both families of states (\(\{ \psi_{AB}^\lambda\}_\lambda\) and \(\{\phi_{AB}^\lambda\}_\lambda\)), is both a necessary and a sufficient condition for the existence of EFI pairs as in Definition~\ref{def: EFI pair}.
\end{theorem}
\begin{proof}
    The proof follows directly from Lemmas~\ref{lemma:efi implies pe ge} and~\ref{lemma: pse imply efi GE}.
\end{proof}

\section{Discussion}
In this work, we investigate the relation between two operational versions of pseudoentanglement, introduced in~\cite{ABV23} and extended in~\cite{GE24}, and EFI pairs~\cite{BCQ23}.
While the landscape of pseudoentanglement definitions is broad, given how recently these definitions were established, we aim to study the versions that entail an operational meaning.
As a central focus of our work, the nature of cryptography and cryptographic protocols carries such an operational meaning, closely related to the LOCC under the computational constraints framework first defined in~\cite{ABV23}.
Moreover, the difference between the studied definitions~\cite{ABV23,GE24} regarding the efficiency of the distillable entanglement of the high-entanglement family further highlights the fundamental difference between efficient and inefficient verification, which has been recurring in the landscape of minimal assumptions for cryptography (e.g., for one-way state generators or one-way puzzles).

In the seminal work of~\cite{ABV23}, the first relations between fully-computational pseudoentanglement and cryptography were established.
Particularly, it was shown how to construct pseudoentanglement from OWFs.
Still, relating pseudoentanglement with weaker primitives in cryptography, such as EFI pairs, was left as an open question in~\cite{ABV23}. 
We answer this question by proving that EFI pairs can be constructed from pseudoentanglement, when the underlying families are uniformly efficiently preparable, and are hence necessary for the existence of pseudoentanglement (Theorem~\ref{thm:pe_to_efi}).
Moreover, it had previously been shown that the existence of a relaxed definition of pseudoentanglement, inefficiently-distillable pseudoentanglement, is necessary for the existence of EFI pairs, with the converse being left as an open question~\cite{GE24}.
Here, we also respond to this question affirmatively, showing that the existence of EFI pairs is both necessary and sufficient for the existence of efficiently preparable but inefficiently-distillable pseudoentanglement (Theorem~\ref{thm:ipe_equiv_efi}).

In proving these theorems, we introduce the study of the continuity of computational entanglement by establishing a technical result demonstrating the error-continuity of the computational entanglement cost (Lemma~\ref{lemma: continuity cost to the same epr}).
This is the first continuity relation to be shown regarding computational entanglement measures.
Moreover, we show that an entanglement gap between the computational distillable entanglement and cost of two families of states depends on a large trace distance between the states of those same families (Lemma~\ref{lemma: connection of td and cd}).
Additionally, we introduce a binary quantum state discrimination task that asks to distinguish between states sampled uniformly from one of two families, and we prove a sufficient bound on the trace distance between mixtures of states of respective families (Lemma~\ref{lemma: td of mixtures}). 
This lemma demonstrates a general amplification result that is exponential in the number of copies, in line with the usual amplification results in cryptography.
We believe that these newly established connections advance the understanding of the role of physical pseudoresources in computational cryptography, particularly in studying its minimal assumptions.
Furthermore, the techniques we developed for the proofs of these relations may also be of independent interest, with applications not only in cryptography but also in other fields of quantum information theory.

Through a more physical lens,~\cite{ABV23} discusses how pseudoentanglement reflects properties of the AdS/CFT correspondence.
Exploring this connection to bridge results from physics and computational complexity, and vice versa, could prove a relevant endeavor.
There, another path would be to link the existence of computational cryptography to emerging phenomena arising from fundamental hypotheses in high-energy physics.
Moreover, our results also fit in a wide realm connecting efficient computation and physics, from the strong Church-Turing thesis to more recent results exploring how computational restrictions affect entanglement theory~\cite{ABV23}, including new limit laws for pure-state computational entanglement~\cite{LRIJ25}, properties of computational entanglement measures~\cite{RLE25}, and theories of efficient measurements \cite{YHK25}. 

Lastly, our work leaves important open questions. 
First, it remains open whether it is possible to build more powerful cryptographic primitives from fully-computational pseudoentanglement, as our work does not establish how powerful fully-computational pseudoentanglement actually is.
Nevertheless, inefficiently-distillable pseudoentanglement is as weak as EFI pairs.
Showing the equivalence or a degree of separation between the different pseudoentanglement definitions is still a relevant question.
Independently, relating physical phenomena and hypotheses to computational hardness conjectures, as well as analyzing the feasibility and experimental realizations of either version of pseudoentanglement, are also interesting directions on the physical side.

\bibliographystyle{halpha}
\bibliography{bib}

\end{document}